\documentclass{article}

\usepackage{hdb_macros}
\usepackage{fullpage}

\usepackage[multiuser,inline,nomargin]{fixme}
\usepackage[style=alphabetic,natbib=true]{biblatex} %
\FXRegisterAuthor{h}{eh}{\color{brown}Huck}
\FXRegisterAuthor{a}{ea}{\color{magenta}Amir}
\FXRegisterAuthor{e}{ee}{\color{blue}Evelyn}
\FXRegisterAuthor{m}{em}{\color{violet}Mitchell}
\fxusetheme{color}

\newcommand{\fullornot}[2]{#1}
\newcommand{\full}[1]{#1}

\newcommand{\res}{R} %

\title{Graph Inference with Effective Resistance Queries}
\author{Huck Bennett\thanks{University of Colorado Boulder. \email{huck.bennett@colorado.edu}.} \and Mitchell Black\thanks{University of California San Diego. \email{m4black@ucsd.edu}.} \and Amir Nayyeri\thanks{Oregon State University. \email{Amir.Nayyeri@oregonstate.edu}.} \and Evelyn Warton\thanks{Oregon State University. \email{wartone@oregonstate.edu}}}
\date{\today}
\begin{document}
\maketitle
\listoffixmes

\begin{abstract}
The goal of \emph{graph inference} is to design algorithms for learning properties of a hidden graph using queries to an oracle that returns information about the graph. Graph reconstruction, verification, and property testing are all special cases of graph inference.

In this work, we study graph inference using an oracle that returns the \emph{effective resistance} (ER) distance between a given pair of vertices. 
Effective resistance is a natural notion of distance that arises from viewing graphs as electrical circuits, and  has many applications. 
However, it has received little attention from a graph inference perspective. Indeed, although it is known that an $n$-vertex graph can be uniquely reconstructed by making all $\binom{n}{2} = \Theta(n^2)$ possible ER queries, very little else is known.
We address this and show a number of fundamental results in this model, including:
\begin{enumerate}
\item $O(n)$-query algorithms for testing whether a graph is a tree; deciding whether two graphs are equal assuming one is a subgraph of the other; and testing whether a given vertex (or edge) is a cut vertex (or cut edge).
\item Property testing algorithms, including for testing whether a graph is vertex-biconnected and whether it is edge-biconnected. We also give a reduction that shows how to adapt property testing results from the well-studied bounded-degree model to our model with ER queries. This yields ER-query-based algorithms for testing $k$-connectivity, bipartiteness, planarity, and containment of a fixed subgraph.
\item Graph reconstruction algorithms, including an algorithm for reconstructing a graph from a low-width tree decomposition; a $\Theta(k^2)$-query, polynomial-time algorithm for recovering the entire adjacency matrix $A$ of the hidden graph, given $A$ with $k$ of its entries deleted; and a $k$-query, exponential-time algorithm for the same task.
\end{enumerate}
We additionally compare the relative power of ER queries and shortest path queries, which are closely related and better studied. Interestingly, we show that the two query models are incomparable in power.
\end{abstract}

\section{Introduction}
\label{sec:intro}

Let $G = (V,E)$ be a hidden graph with $n$ vertices, and let $d: V \times V \rightarrow \R^{\geq 0}$ be a metric on $G$.\fullornot{}{\footnote{Due to space constraints, this version of the paper omits a number of results and proofs. We strongly encourage the reader to view the full version of our paper, which we have submitted as a supplementary file. The full version contains all discussion, results, and proofs.}}
We study \emph{graph inference} problems, where the goal is to reconstruct $G$ or deduce properties of $G$ using a limited number of queries to $d$. In this work, we focus on the case where $d$ is the \emph{effective resistance} (ER) metric on $G$.

Graph inference using the ER metric and its close cousin, the \emph{shortest path} (SP) metric, arises naturally in many applications. Inferring a graph from SP queries has been widely studied in the context of network discovery~\citep{journals/jsac/BeerlixovaEEHHMR06}, where it is used to map unknown regions of the internet. It is also applied in the reconstruction of evolutionary trees~\citep{hein1989optimal,conf/soda/KingZZ03,journals/ipl/ReyzinS07}, where understanding evolutionary relationships based on pairwise distances is crucial. On the other hand, ER-based inference is particularly relevant in the study of social network privacy~\citep{book/affarwal2014social_network_data,Hoskins2018Inferring}, where the goal is to assess the potential of revealing hidden network features through random-walk-based queries.

Graph inference via SP queries has also been explored extensively from a theoretical perspective, leading to notable algorithmic and lower bound results. For instance, bounded-degree graphs can be reconstructed using $\widetilde{O}(n^{3/2})$ queries~\citep{journals/talg/KannanMZ18}, and bounded-degree almost-chordal graphs can be reconstructed using $O(n \log n)$ queries~\citep{journals/arxiv/bestide2023reconstnolongcycle}. On the other hand, $\Omega(n^2)$ queries are needed to even reconstruct trees if they are not bounded degree~\citep{alt/ReyzinS07}. However, the study of graph inference via ER queries is still in its early stages. Beyond the $\Theta(n^2)$-query algorithm that guarantees full reconstruction given all pairwise ER distances (see~\citep{journals/tcs/WittmannSBT09,Spielman2012TreesRecNotes,Hoskins2018Inferring}), little is known that comes with theoretical guarantees. Moreover, the quadratic query complexity needed to reconstruct a whole graph is impractical for large graphs and when the ability to make ER measurements is limited. The focus of this work is to address the corresponding natural question:

\begin{quote}
What properties of a graph can be inferred from a subquadratic number of ER queries?
\end{quote}
We interpret this question quite generally, and study graph reconstruction, property testing, and more using ER queries. 

A closely related question asks what the comparative power of SP and ER queries is in the context of machine learning on graphs. In the context of graph neural networks and graph transformers, distance measures such as SP and ER are often used as inputs to the network in the form of positional encodings~\citep{zhang2023rethinking,velingker2024affinity}. A primary challenge in graph neural networks is to identify which positional encodings offer the highest expressive power, enabling the model to detect the most properties of a graph. For example, graph transformers using ER as a positional encoding can determine which vertices are cut vertices, while graph transfomers using SP cannot~\citep{zhang2023rethinking}. While there have been several works aiming to understand and compare the relative expressive power of SP and ER encodings~\citep{zhang2023rethinking, black2024comparing, velingker2024affinity}, the capabilities of SP or ER remain only partially understood. A deeper understanding of their potential could lead to better feature selection strategies and improved learning efficiency.

\subsection{Our Contribution}
In this paper, we study a number of fundamental graph inference problems using ER queries.
To begin with, we demonstrate that with $O(n)$ ER queries, the following tasks can be performed:

\begin{enumerate}[(1)]
    \item Determining whether a graph is a tree \fullornot{(\Cref{thm:tree-verification})}{(see full version)}.
    \item Determining whether two graphs are identical, assuming one is a subgraph of the other. In other words, we can detect if some edges have been added to the graph or if some edges have been deleted, but not both~(\Cref{thm:proper-subgraph-verification}).
    \item Deciding whether a given vertex is a cut vertex~(\Cref{thm:cut-vertex-verification}).
    \item Deciding whether a pair of vertices are adjacent with a cut edge\fullornot{~(\Cref{thm:cut-edge-verification})}{~(\cref{lem:cutedge-iff-res-is-plusorminus-one})}.
\end{enumerate}

Furthermore, we explore (approximate) \emph{property testing} of graphs, where our algorithms can distinguish between the case where the graph satisfies a given property and the case where it is $\eps$-far from satisfying the property, meaning that at least $\eps m$ edge modifications are required to satisfy the property. In this setting, we present the following results:

\begin{enumerate}[(1)]
    \item Determining whether a graph is vertex-biconnected or $\eps$-far from being vertex-biconnected with $O(n/\eps)$ ER queries\fullornot{~(\Cref{thm:vertex_bi_conn_test})}{ (see full version)}.
    \item Determining whether a graph is edge-biconnected or $\eps$-far from being edge-biconnected with $O(n/\eps^2 + 1/\eps^4)$ ER queries\fullornot{~(\Cref{thm:edge_bi_conn_test})}{ (see section 3.2)}.
    \item Showing that for any property that can be tested in $f(\eps, n)$ time in the well-studied bounded-degree model, there is an algorithm to test it for bounded-treewidth graph with $f(\eps, n) \cdot n$ ER queries (\Cref{thm:adj_list_pt_to_er_pt}). 
    As a result, we obtain property testers that require a subquadratic number of ER queries for bounded-treewidth graphs, addressing a range of problems such as the inclusion of a fixed subgraph, $k$-connectivity, bipartiteness, the presence of long cycles, and any minor-closed property (e.g., planarity)~\cite{books/cu/Goldreich17}.

\end{enumerate}

We remark that the running time of our property testing algorithms for the latter two results can be reduced by a factor of $n$, making them dependent only on $\eps$, given a stronger ER query model. In this model, queries return vertices in the order of their distance from a vertex $v$, allowing the algorithm to request only the closest $k$ vertices to $v$. This stronger query model is natural, as in most applications, closer vertices are often accessible with less effort, for example if close vertices are sampled by a random walk.
\\ \\
In addition to our property inference algorithms, we analyze the relative power of ER and SP query models. Interestingly, we show that the two query models are incomparable, i.e., neither model is strictly stronger than the other for all tasks. Specifically, we demonstrate the following results:

\begin{enumerate}[(1)]
    \item Determining whether a graph is a clique can be achieved with $O(n)$ ER queries, whereas it requires $\Omega(n^2)$ SP queries (\Cref{thm:er-sp-clique}).
    \item Checking whether two given vertices are adjacent can be achieved with a single SP query, whereas it requires $\Omega(n)$ ER queries (\Cref{thm:er-sp-adjacency}).
\end{enumerate}

Finally, we study the problem of graph reconstruction using ER queries. Specifically, we give an algorithm for each of the following tasks:
\begin{enumerate}[(1)]
\item Given a width-$k$ tree decomposition of a graph $G$, reconstructs $G$ using $O(k^2 n)$ ER queries \fullornot{(\cref{thm:reconstruction-from-tree-decomp})}{ (see full version)}.
\item Given an adjacency matrix $A$ of a graph $G$ with $k$ missing entries, recovers $A$ using $O(k^2)$ queries and runs polynomial time \fullornot{(\cref{thm:k-unknowns-quadratic_queries})}{ (see full version)}.
\item Given an adjacency matrix $A$ of a graph $G$ with $k$ missing entries, recovers $A$ using $k$ queries and runs in exponential time (\cref{thm:k-unknowns-exponential-completion}).
\end{enumerate}

Our reconstruction algorithm in \cref{thm:k-unknowns-exponential-completion} uses techniques from convex analysis and relies on a highly nontrivial structural property of ER distances, which allows us to reduce the number of queries. It remains an open question whether a polynomial-time algorithm can be achieved with a subquadratic number of queries, \fullornot{i.e., whether it is possible to combine the strengths of \cref{thm:k-unknowns-quadratic_queries,thm:k-unknowns-exponential-completion}}{that is, whether quadratic Schur complement reconstruction and \cref{thm:k-unknowns-exponential-completion} can be combined to give a low query complexity and runtime complexity}.

\subsection{Related work}
\paragraph{Graph inference with different query models.}
Graph inference has been studied under various models, with {edge detection} and {edge counting} being two prominent approaches motivated by applications in biology. In these models, queries allow one to check whether an induced subgraph contains any edges or to determine the number of edges in the subgraph~\citep{sicomp/AlonBKRS04, journals/jcss/AngluinC08, journals/siamdm/AlonA05, conf/wg/BouvelGK05}.~\citet{alt/ReyzinS07} provides an extensive survey of results within these models, as well as the shortest path (SP) query model. Notably, they show that reconstructing a hidden tree requires at least $\Omega(n^2)$ SP queries, establishing that the bounded-degree assumption is necessary for obtaining nontrivial results using SP queries.

\citet{journals/talg/KannanMZ18}
demonstrated that bounded-degree graphs can be fully reconstructed using $\widetilde{O}(n^{3/2})$ SP queries. They further showed that bounded-treewidth chordal graphs and outerplanar graphs can be reconstructed with $\widetilde{O}(n\log n)$ SP queries. This result has been recently generalized to bounded-treewidth graphs without long cycles by \citet{conf/iwpec/Bestide24}, with additional related work by \citet*{journals/tcs/RongLYW21}.

In contrast to the SP model, significantly fewer theoretical results are known for the {effective resistance (ER) model}. It has been established that a hidden graph can be fully reconstructed if the ER distances between all pairs of its vertices are known~\citep{journals/tcs/WittmannSBT09,Spielman2012TreesRecNotes,Hoskins2018Inferring}. On the other hand, since ER and SP distances are equivalent for trees, it follows that reconstructing a general tree requires $\Omega(n^2)$ ER queries, and hence the known $O(n^2)$ query algorithm is tight for general graphs, as is the case for the SP model. However, in contrast to the SP model, subquadratic algorithms are not known for any family of graphs beyond bounded-degree trees.

A continuous variant of the graph reconstruction problem, known as \emph{Calder\'{o}n's inverse problem}, is to recover the conductivity of an object from measurements of current and potential on its surface.  Calder\'{o}n's inverse problem has been studied extensively by mathematicians~\citep{Uhlmann2012} and found important applications in Electrical Impedance Tomography (EIT) in medical imaging~\citep{uhlmann2009electrical} and Electrical Resistivity Tomography (ERT) in geophysics~\citep{wikipediaERT}.

\paragraph{Property testing.}
We study the gap version of checking properties of graphs using ER queries. In this setting, the goal is to distinguish between graphs that possess a certain property and those that are ``far'' from having that property, using only a few ER queries. This part of our work is related to property testing of graphs in the bounded-degree model. Many properties have been studied in this model, including $k$-connectivity, bipartiteness, subgraph exclusion, minor exclusion, and planarity. We refer the reader to~\citet{books/cu/Goldreich17} for an extensive exposition of known results in this area. 
For contrast, we study property testing algorithms that rely only on resistance distance queries rather than adjacency queries. 

More broadly, our work is also related to sublinear-time algorithms for finite metric spaces. Such algorithms become particularly relevant when one is interested in estimating parameters such as the average ER distance or identifying the central point in the ER metric space. Notably, these algorithms operate without reconstructing the graph or explicitly learning its topology. There is a substantial body of work on testing properties of metric spaces; see, for example,~\cite{conf/stoc/Indyk99,journals/iandc/ParnasR03,conf/icalp/Onak08}.

\paragraph{Graph machine learning.}

One important (albeit perhaps less direct) motivation for our work is graph machine learning. Currently, there are no polynomial-time machine learning algorithms (e.g.,~graph neural networks) that can completely capture the structure of a graph; that is to say, all existing methods will give the same output on some set of non-isomorphic graphs (see~\citep{xu2018how} for an example). This means that no graph learning algorithm captures all properties of a graph.
\par 
As an imperfect solution, one approach is compute some quantities associated with the graph and use these as input to the machine learning algorithm. For example, one could compute the effective resistance between all pairs of vertices (as in~\citep{zhang2023rethinking,velingker2024affinity}) and use this as an input to a neural network. These graph quantities are sometimes called \textit{positional encodings}. Since the use of positional encodings is known not to completely capture the topology of a graph, researchers instead ask which properties of a graph these encodings do capture. For example, effective resistance increases the expressive power of message-passing neural networks~\citep{velingker2024affinity}, transformers~\citep{vaswani2017attention} using effective resistance as a positional encoding can determine which vertices are cut vertices (a property neither message-passing neural networks nor shortest-path distance can detect)~\citep{zhang2023rethinking}, and transformers using effective resistance can determine which edges are cut edges~\citep{black2024comparing}. At a high level, the goal is to understand the power of effective resistance in deducing graph properties when used as a positional encoding in a graph neural network. This question becomes particularly relevant to our work when it focuses on identifying a small subset of effective resistances that are sufficiently informative for the task at hand, thereby reducing computational costs.

\section{Preliminaries}
\label{sec:prelims}

\subsection{Graph Laplacian and Effective Resistance}
Let $G=(V, E, w)$ be an undirected 
graph with edge weights
$w:E\rightarrow \R^{\geq 0}$. We denote the number of vertices as $n:=|V|$ and the number of edges as $m:=|E|$. The \emph{graph Laplacian} is the $n \times n$ matrix defined $L := D - A$, where the \emph{weighted degree matrix} is the diagonal matrix with entries $D_{u,u} = \sum_{(u,v)\in E}{w(u,v)}$ and the \emph{weighted adjacency matrix} is the matrix with entries $A_{u,v} = w(u,v)$ if $(u,v)\in E$ and $A_{u,v}=0$ otherwise.
Alternatively, the graph Laplacian is $L = \partial W\partial^{T}$, where $W$ is the $m \times m$ diagonal weight matrix with entries $W_{e,e} = w(e)$ and $\partial$ is the $n \times m$ \emph{signed incidence matrix} (corresponding to the \emph{boundary operator}), which has  entries\footnote{The choice of which of the first two cases is equal to $1$ and which equal to $-1$ is arbitrary. The graph Laplacian will be the same for either choice, and other quantities involving the signed incidence matrix will only differ by a sign flip.}
\[
\partial_{v, e} := \begin{cases}
1 & \text{if $e = (v, w)$ for some $w$ ,} \\
-1 & \text{if $e = (u, v)$ for some $u$ ,} \\
0 & \text{otherwise .} \\
\end{cases}
\]
Finally, a third way of writing the graph Laplacian is as the sum $L = \sum_{(u,v)\in E} w(u,v)L_{uv}$, where $L_{uv} = (\vec{1}_{u} - \vec{1}_{v})(\vec{1}_{u} - \vec{1}_{v})^{T}$ is the \emph{edge Laplacian}.

    Let $M \in \mathbb{R}^{m \times n}$ be a matrix.
    The \emph{psuedoinverse} of $M$, denoted $M^+$, is any matrix such that the following conditions hold:
    (1) $AA^+A=A$,
    (2) $A^+AA^+=A^+$,
    (3) $(AA^+)^T=AA^+$, and 
    (4) $(A^+A)^T=A^+A$.

The \emph{effective resistance} between a pair of nodes $u$ and $v$ is defined as
\begin{equation} \label{eq:eff_res_def}
\res(u,v) := (\vec{1}_u - \vec{1}_v)^{T} L^{+} (\vec{1}_u-\vec{1}_v) \ \text{,}
\end{equation}
where $L^{+}$ is the pseudoinverse of $L$ and $\vec{1}_u,\vec{1}_v \in \R^n$ are the indicator vectors of $u$ and $v$. One can show that $\res$ is a metric on the vertices in the graph.
Interpreting graphs as electrical circuits, the effective resistance measures the resistance between $u$ and $v$ in the electrical network where each edge $e$ has conductance equal to its weight, $w(e)$, giving each edge a resistance of $1/w(e)$. %
The \emph{$uv$-potentials} are $\Vec{p}_{u,v} := L^{+} (\vec{1}_u-\vec{1}_v)$, as these are the voltage potentials on the vertices resulting in a unit flow of current from $u$ to $v$ in the graph.

Notably, the kernel of any graph Laplacian $L$ contains the span of the all-ones vector (i.e., $\lspan(\vec{1}) \subseteq \ker(L)$), since the weighted degree of a vertex is equal to the sum of the weights of its incident edges. Moreover, if the graph is connected then $\ker(L) = \lspan(\vec{1})$.

To see the matrix form of the relationship between $L^+$ and effective resistance, let $J$ be the all-ones matrix. The matrix $I - \frac{1}{n}J$ projects any vector onto the orthogonal complement of $\lspan(\vec{1})$.
Its conjugation with $R$ be the matrix of all effective resistances with $R_{i,j}=\res(i,j)$, gives $L^+$ up to a constant factor:
\begin{equation}
    \label{eq:eff-res-matrix-form}
     -\frac{1}{2}\left(I-\frac{1}{n}J\right)R \left(I-\frac{1}{n}J\right)=L^+
\end{equation}

This formula implies the following lemma~\citep{journals/tcs/WittmannSBT09, Spielman2012TreesRecNotes, Hoskins2018Inferring}.
\begin{lemma}
\label{lem:full_reconstruciton}
    Any weighted graph with $n$ vertices can be reconstructed with $\binom{n}{2}$ effective resistance queries.
\end{lemma}
\full{
Due to the $O(n^\omega)$ time algorithm for obtaining a matrix's rank-normal form given by~\citep{KellerGehrig1985}, the psuedoinverse of $L$, and hence the matrix of all effective resistances, can be obtained in $O(n^\omega)$ runtime where $\omega<2.371339$~\citep{ADWXX2024} is the matrix multiplication exponent.
}

\paragraph{The regularized graph Laplacian}
As mentioned above, the graph Laplacian is never full-rank regardless of the graph; the all-ones vector $\vec{1}$ is always in the kernel. The \emph{regularized graph Laplacian} is a slightly modified version of the graph Laplacian defined $\widetilde{L} = L +  \frac{1}{n} \cdot J$. For a connected graph, $\widetilde{L}$ is full-rank. 

\begin{lemma}[{\citet[Equation 9]{ghosh2008minimizing}}]
    Let $G = (V, E)$ be a connected graph with Laplacian $L$. The following two facts are true of the regularized graph Laplacian: (1) $\tilde{L}^{-1} = L^{+} + \frac{1}{n} J$,
        and (2), for all $u, v \in V$, $\res(u,v) = (\vec{1}_u - \vec{1}_v)^{T} \tilde{L}^{-1} (\vec{1}_u-\vec{1}_v)$.
\end{lemma}

Alternatively, we can characterize the effective resistance in terms of the amount of the amount of current on the edges. A \emph{unit uv-flow} is
a flow $f:E\to\R$ such that $\partial f = \vec{1}_u - \vec{1}_v$. The  \emph{energy} of a flow $f:E\rightarrow \R$ is defined to be %
$\mathcal{E}(f):=\sum_{e\in E}{f(e)^2/w(e)}$. 

\begin{lemma}[Thompson's Principle~\citep{thomson1867treatise}]
\label{lem:thompson_dirichlet}
Let $\mathcal{F}_{uv}$ be the set of unit $uv$-flows. The effective resistance between two vertices $u,v\in V$ is equal to the minimum energy of any unit $uv$-flow, $\res(u,v)= \min_{f \in \mathcal{F}_{uv}} \mathcal{E}(f)$.
\end{lemma}
The minimum energy $uv$-flow $f_{u,v}:= \argmin_{f \in \mathcal{F}_{uv}} \mathcal{E}(f)$ is called the \emph{electrical $uv$-flow}.
Thompson's principle immediately implies \emph{Rayleigh's monotonicity law}: if the resistance values increase the effective resistance between any pair of vertices can only increase. 

\begin{lemma}[Rayleigh's Monotonicity Law~\citep{rayleigh1877theory}]
    Let $G=(V, E_G,w_G)$ and $H=(V, E_H, w_H)$ be weighted graphs on the same set of vertices $V$. If $E_G\subseteq E_H$ and $w_G(e) \leq w_H(e)$ for all $e\in E_G$, then $\res_G(u,v) \geq \res_H(u,v)$ for all $u,v\in V$. 
\end{lemma}
We find the following basic properties of effective resistance useful in this paper.
\begin{lemma}
\label{lem:er_and_edge_cuts}
Let $G=(V,E)$. Then for all $u, v \in V$,
\begin{enumerate} [(1)]
    \item \label{lem:er_of_an_edge}
    If $(u,v)\in E$, then $\res(u,v) \leq 1$.  Moreover, $\res(u,v) = 1$ if and only if $(u,v)$ is a cut edge.
    \item \label{lem:er_of_path}
    If there is exactly one path $\gamma$ between two vertices $u,v\in V$, then the $\res(u,v)$ is the length of $\gamma$.
    \item The ER distance between any pair of vertices in two different edge biconnected components is at at least one.
\end{enumerate}
\end{lemma}

\paragraph{Schur complement.} The \emph{Schur complement} of the Laplacian $L$ of a matrix on a vertex subset $U$, denoted $L_U$, is defined as follows.\footnote{The Schur complement is usually defined as $L_U := L(U, U) - L(U,\overline{U})L(\overline{U}, \overline{U})^{-1}L(\overline{U}, U)$ (e.g., in~\citep[Lemma 11.5.3]{spielman2019sagt}), using the inverse $L(\overline{U}, \overline{U})^{-1}$ instead of the pseudoinverse. We use the pseudoinverse as it is more general while preserving the relevant properties of the Schur complement. To see why it preserves these properties, note that the submatrix $L(\overline{U}, \overline{U})$ is invertible iff there is an edge between $U$ and $\overline{U}$. However, if there is no edge between $U$ and $\overline{U}$, then $L(U,\overline{U})=0$ and $L_U = L(U,U)$. Likewise, if there are no edge between $U$ and $\overline{U}$, then $L_U$ is a graph Laplacian of subgraph of $G$ induced by $U$, so all the properties of the Schur complement mentioned in this section hold.}

\begin{equation}
\label{eqn:schur_comp}
L_U := L(U, U) - L(U,\overline{U})L(\overline{U}, \overline{U})^+L(\overline{U}, U) \ \text{,}
\end{equation}
where $L(X,Y)$ is the submatrix of $L$ indexed by rows in $X$ and columns in $Y$ for $X, Y \subseteq V$. 
Note that for a connected graph, $L(\overline{U},\overline{U})$ is always full rank.
It is known that the Schur complement of a Laplacian of a graph $G$ is the Laplacian of a (possibly weighted) graph with vertex set $U$, which we refer to as $G_U$.  The following lemma summarizes some properties of the Schur complement that we find helpful in this paper~\citep{spielman2019sagt}.
\begin{lemma}
\label{lem:schur_comp_basics}
Let $G=(V,E)$ be a weighted graph, let $U\subseteq V$, and let $G_U$ be the Schur complement of $G$ on $U$.  Then the following properties hold.
\begin{enumerate}[(1)]
    \item For any $u,v\in U$, their effective resistance in $G$ is equal to their effective resistance in $G_U$.

    \item Let $u\in U$ be such that $w_G(u,v) = 0$ for all $v \in \overline{U}$. Then, for any $u'\in U$, $w_G(u,u') = w_{G_U}(u,u')$, i.e., the neighborhood of $u$ is identical in $G$ and $G_U$.
\end{enumerate}
\end{lemma}

The Schur complement matrix can be obtained by using Gaussian elimination on the graph Laplacian. Likewise, for each step of Gaussian elimination, there is a corresponding operation on the graph, with the operations corresponding to entire process of Gaussian elimination resulting in the graph $G_U$. For a vertex $v\in\overline{U}$, performing Gaussian elimination on the $v$th row and column corresponds to removing the vertex $v$ and, for each pair of neighbors $(a,b) \in N(v) \times N(v)$, replacing its weight with $w(a, b)' := \frac{w(a,b)}{d(v)}$, the ratio between the edge weight and the degree of $v$.

\paragraph{Tree decomposition.} Let $G=(V,E)$ be a graph. A \emph{tree decomposition} is a pair $({\cal B}, T)$, where ${\cal B}$ is a set of subsets of $V$ called \emph{bags} and $T$ is a tree with vertex set ${\cal B}$ such that
\begin{enumerate}[(1)]
    \item $\bigcup_{B\in{\cal B}} B = V$.
    \item For every $\{u,v\}\in E$, there exists $B\in{\cal B}$ such that $u,v\in B$.
    \item Let $B, B''\in {\cal B}$, and suppose $B'\in{\cal B}$ is a bag on the unique $B$ to $B''$ path in $T$.  Then, $B\cap B''\subseteq B'$.
\end{enumerate}

We show that given a tree decomposition of a graph with bounded treewidth, it is possible to fully recover the graph using $o(n^2)$ effective resistance queries.
\paragraph{Graph cuts and connectedness.}
Let $G=(V,E)$ be a graph.
A vertex $v \in V$ is a \emph{cut vertex} if the subgraph induced by $V \setminus \set{v}$ has more connected components than $G$.
An edge $(u,v) \in E$ is a \emph{cut edge} if the subgraph given by $(V, E \setminus \{(u,v)\})$ has more connected components than $G$.
A graph is \emph{vertex biconnected} if it contains no cut vertices.
A graph is \emph{edge biconnected} if it contains no cut edges.

\section{Checking graph properties}
\label{sec:prop-testing}

\fullornot{
There are several properties of graphs that can be characterized by the resistance distance on a subset of the graph's vertices.
One such property is whether a given vertex is a cut vertex.
We give a characterization for such vertices, and use this result to solve several decision problems on graphs, as well as to give property testing algorithms for certain connectedness properties, which, in a sense, gives an approximate solution to the corresponding decision problems.
\begin{lemma}
    \label{lem:vertex-cut-iff-triangle-ineqaulity-tight}
    Let $G=(V,E)$ be an undirected, unweighted graph, and let $a,b,c \in V$ be distinct vertices. 
    Then, $b$ is a cut vertex between $a$ and $c$ (that is, $a$ and $c$ are in different connected components of
    the subgraph of $G$ induced by $V \setminus \{ b \}$, $G_{| V \setminus \set{b}}$) if and only if $\res(a,b) + \res(b,c) = \res(a,c)$.
\end{lemma}
\begin{proof}
    First we will prove that the triangle inequality holds exactly equal when the middle vertex is a cut vertex.
    Suppose $b$ is a cut vertex between $a$ and $c$, and $f \in \mathbb{R}^{|E|}$ be the unit electrical flow.
    By the flow decomposition $\exists f_1, \ldots, f_k \in \mathbb{R}^{|E|}$ such that $f = f_1 + \cdots + f_k$ for some $k \leq |E|$.
    Moreover, for $i \leq k$, there is a simple $a$ to $c$ path $\pi_i$ that is the support of $f_i$.
    For each path $\pi_i$, since $b$ is a cut between $a$ and $c$, it follows that $b$ must lie on this path.
    
    Let $V \setminus {b} = A \cup C$ be a disjoint partition of vertices other than $b$, by assumption, and denote $G^A$ to be the subgraph of $G$ induced by the vertex set $V \setminus C$ and $G^C$ to be the subgraph of $G$ induced by the set $V \setminus A$.
    If we restrict $f_i$ to $G^A$ then $b$ will be the endpoint of each path $\pi_i|_{G^A}$.
    This is because $c$ lies on the other side of the cut vertex $b$, and a simple path cannot traverse $b$ twice.
    So it follows that this is a unit $a$ to $b$ flow.
    If a lower energy unit $a$ to $b$ flow existed in $G^A$, it could be used in place of $f$ under this restriction, reducing the overall energy of $f$ (this would contradict that $f$ is the electrical flow).
    A lower energy unit $a$ to $b$ cannot exist in $G$ but not in $G^A$, since everything on $G^C$ would clearly form a circulation back to $b$.
    So restricting $f$ to $G^A$ gives a miminum energy $a$ to $b$ flow, and a similar arugment gives the fact that restricting $f$ to $G^C$ gives a mimimum energy $a$ to $c$ flow.

    Since the edges of $G^A$ and $G^C$ are disjoint it follows that $\|f|_{G^A}\|_2^2 + \|f|_{G^C}\|_2^2 = \|f\|_2^2$, or equivalently, that \[
        \res(a,b) + \res(b,c) = \res(a,c).
    \]

    Suppose now, to show the other direction of implication, that $\res(a,b) + \res(b,c) = \res(a,c)$.
    By definition this is equivalent to having
    \[
        (\Vec{1}_a - \Vec{1}_b)^T L^+ (\Vec{1}_a - \Vec{1}_b) + (\Vec{1}_b - \Vec{1}_c)^T L^+ (\Vec{1}_b - \Vec{1}_c) = (\Vec{1}_a - \Vec{1}_c)^T L^+ (\Vec{1}_a - \Vec{1}_c).
    \]
    With some rearranging and cancellations, we obtain the equality \[
        \Vec{1}_a L^+ (\Vec{1}_b - \Vec{1}_c) = \Vec{1}_b L^+ (\Vec{1}_b - \Vec{1}_c).
    \]
    Since $L^+ (\Vec{1}_b - \Vec{1}_c)$ gives the potentials in the electrical $b$ to $c$ flow, this means that $a$ and $b$ have equal potential in this flow.
    Let $A \subseteq V$ be the set of all vertices with potential equal to $a$, and let $C = V \setminus A$ be its complement.
    By the equation we obtained, we know $b \in A$.
    Clearly $c$ has potential different than $b$, as they are distinct vertices, and the effective resistance between them, as a metric, must be non-zero.
    Suppose, by contradiction, that $b$ is not a cut vertex between $a$ and $c$.
    Then there must be an edge between $A$ and $C$ whose endpoint in $A$ is a vertex other than $b$, call this vertex $x$.
    This would imply, by the potential difference, that some amount of flow belongs on this edge in the electrical $b$ to $c$ flow.
    But flow exiting $x$ implies flow entering $x$, which implies a vertex exists with potential larger than $b$, the source, yielding a contradiction.
\end{proof}
}
{We study checking properties of the graph with subquadratic number of ER queries.  Specifically, in this section, we show that we can test the equality of two subgraphs, provided that one is a subgraph of the other.  We further show results for testing a given vertex is a cut vertex or a given edge is a cut edge.  We refer the reader to the full version for more discussion on these results and one more result on verifying acyclicity of a connected graph with $O(n)$ ER queries.}

\full{
\subsection{Acyclicity}

\Cref{lem:vertex-cut-iff-triangle-ineqaulity-tight} can be used directly to check whether a vertex is a cut vertex, but it also allows for a concise proof of the following lemma, which states that, in an unweighted graph, the set of distances to any particular vertex must be all integers, if the graph contains no cycles.

\begin{lemma}
    \label{lem:cycle-iff-nonintegral}
    Let $G=(V,E)$ be a connected unweighted graph and $v \in V$ be any vertex.
    The graph $G$ has a cycle if and only if there exists a vertex $u\in V$ such that $\res(u,v)$ is not an integer.
\end{lemma}
\begin{proof}
    If $G$ is acyclic then there is a unique path between every pair of its vertices.  Hence, by \Cref{lem:er_of_path} of \cref{lem:er_and_edge_cuts}, the distance between every pair of vertices is an integer.

    To prove the other direction, suppose $G$ contains a cycle.    
    Let $v \in V$ be arbitrary, and let $x\in V$ be a closest vertex to $v$ that is contained in a cycle $C$.  Furhter, let $(x,y)\in E$ be an edge of $C$ (note $x$ might be equal to $v$).  By \cref{lem:er_of_path} of \cref{lem:er_and_edge_cuts}, $R(v,x)$ is an integer.  By \cref{lem:er_of_an_edge} of \cref{lem:er_and_edge_cuts}, $R(x,y)<1$.  So, by triangle inequality, $R(v,y)< R(v,x)+1$.  On the other hand, by \cref{lem:vertex-cut-iff-triangle-ineqaulity-tight}, $R(v,y) = R(v,x) + R(x,y) > R(v,x)$.  Overall, $R(v,x) < R(v,y)< R(v,x)+1$, and $R(v,x)\in\Z^+$ implies that $R(v,y)$ is not integer.
\end{proof}

It is noteworthy that the characterization holds for any vertex of the graph, giving rise to our algorithm.
By choosing an arbitrary vertex of our graph, and querying its pairwise distance to the remaining vertices, we are able to detect a cycle within our graph using only $n - 1$ queries.

\begin{theorem}
    \label{thm:tree-verification}
    Let $G=(V,E)$ be an undirected unweighted graph, and let $n=|V|$. There exists an algorithm that decides whether $G$ is a tree using $n - 1$ ER queries.
\end{theorem}
\begin{proof}
    Let $v \in V$ be an arbitrary vertex of $G$.
    Verifying its pairwise distances are all integers can be done using $n-1$ ER queries by querying its distance to each other vertex $u \in V \setminus \{ v \}$.
    If all of these ER distance are finite we conclude that $G$ is connected.  If, in addition, all of them are integers we conclude that $G$ is acyclic by \cref{lem:cycle-iff-nonintegral}.  If $G$ is both connected and acyclic then it is a tree.
\end{proof}

The effective resistance, as a metric, is highly sensitive to global changes within the graph, which is why querying all of the distances from only one vertex allows for detecting cycles. This is something that is not possible using shortest path queries.
In fact, the technique of querying the resistance distance of one vertex to each other vertex can be leveraged even further to detect changes in the graph, as we do in the next section. 
}

\fullornot{
\subsection{Equality of graphs}
}{
\paragraph{Equality of graphs.}
}
\fullornot{
The sensitivity of resistance distance can be felt by every vertex of the graph.
Even a small local change, such as the addition of a single edge between two vertices (or even a change to a single edge's weight), has ripple effects throughout the entire graph.
In fact, Rayleigh's Monotonicity Principle states that decreasing the resistance of any edge cannot cause the effective resistance between any two nodes to increase.
What we show is that, for any connected graph, for any vertex $s$ in this graph, adding an edge (or decreasing its resistance for weighted graphs) causes a strict decrease in the resistance distance between $s$ and some other vertex.
In a sense, this gives a `strict monotonicity' property for connected graphs, wherein lowering resistance on any edge guarantees there is a pair involving every vertex for which the effective resistance on that pair decreases.
\begin{lemma}
    \label{lem:connected-strict-monotonicity}
    Let $G=(V,E_G,w_G)$ and $H=(V,E_H,w_H)$ be two connected 
    $n$-vertex graphs such that for every pair $(a,b) \in V \times V$, $w_G(a,b) \leq w_H(a,b)$ and there exists $(u,v) \in V \times V$ such that $w_G(u,v) < w_H(u,v)$ strictly.
    Then for every $s \in V$ there exists a vertex $t \in V$ such that $\res_H(s,t) < \res_G(s,t)$. 
\end{lemma}
\begin{proof}
    By Rayleigh's monotonicity law, it suffices to prove the given statement for the following case:
    suppose that all pairs other than $(u,v)$ have equal weights, and only $w_G(u,v) < w_H(u,v)$.
    
    Let $s\in V$ be any vertex, and suppose without loss of generality that there exists an $s$ to $u$ path $\pi$ that avoids $v$.  This assumption is alright as at most one of the following conditions holds: (1) $u$ is a cut vertex that separates $s$ from $v$, and (2) $v$ is a cut vertex that separates $s$ from $u$.
    
    Now, let $f:E_G \to \mathbb{R}$ be the electrical $s$ to $v$ flow in $G$, and let $f':E_H\to\R$ be the same flow in $H$, i.e..
    \[
        f'(e) = \begin{cases}
            f(e) & e \in E_G \\
            0 & \text{otherwise}
        \end{cases}
    \]
    Note, by Thompson's law, $\res_G(s,v) = \mathcal{E}(f)$, and $\res_H(s,v) \leq \mathcal{E}(f')$ as $f'$ is not necessarily an electrical flow.  We show by considering two cases that in fact $\res_H(s,v) < \mathcal{E}(f')$.
    
    First, we consider the case that $f(u,v) = 0$ (in particular that can happen if $w_G(u,v) = 0$.)
    We say that a vertex $x\in V$ is used by $f'$ if some edge incident to $x$ has a non-zero $f'$ value.  Let $r\in V$ be the closest vertex to $u$ on $\pi$ that is used by $f'$ (note $r$ can be equal to $s$ or $u$).  Then, let $\pi' = \pi[r,u]\circ (u,v)$ be the $r$ to $v$ path in $H$, and note that none of the edges of $\pi'$ is used in $f'$.  Also, since $r$ is used by $f'$ there is at least one $r$ to $v$ path $\gamma'$ in $f'$ that has positive $f'$ value on all its edges.  Now, we build $f''$ in $H$ by sending an $\eps$ amount of flow on $\pi'$, and pushing back an $\eps$ amount of flow on $\gamma'$, i.e.
    \[
        f''(e) = \begin{cases}
            \eps & e \in \pi' \\
            f'(e)-\eps & e \in \gamma' \\
            f'(e) & \text{otherwise.}
        \end{cases}
    \]
    It follows that,
    \[
    g(\eps) = \mathcal{E}(f'') - \mathcal{E}(f') = \sum_{e\in\pi'}{(w_H(e)^{-1} \cdot \eps)^2} +\sum_{e\in\gamma'}{ \left((w_H(e)^{-1}\cdot(f'(e)-\eps))^2 - (w_H(e)^{-1}\cdot f'(e))^2\right)}.
    \]
    Thus,
    \[
    \frac{dg}{d\eps}=2\varepsilon\sum_{e\in\pi'}{(w_H(e)^{-1})} - 2\sum_{e\in\gamma'}{w_H(e)^{-1}\cdot(f'(e)-\eps)}\]
    which, when evaluated at $\varepsilon = 0$, is $-2 \sum_{e \in \gamma'}w_H(e)^{-1} \cdot f'(e) < 0$.
    Since $g(0) = 0$, and the derivative of $g$ is negative at $0$, for sufficiently small values of $\eps$, $g(\eps) < 0$, which means $\mathcal{E}(f'') < \mathcal{E}(f')$.  But, we already know $\res_H(s,v) \leq \mathcal{E}(f'')$ and $\res_G(s,v) = \mathcal{E}(f')$, implying, in the case that $w_G(u,v)=0$, the proof is complete.

    Second, we consider the case that $f(u,v) > 0$.  In this case, we have $w_G(u,v) > 0$, as otherwise $f$ cannot use the non-existent edge.
    Then trivially \[\mathcal{E}(f) = \sum_{e \in E_G} (w_G(e)^{-1} \cdot f(e))^2 >  \sum_{e \in E_G} (w_H(e)^{-1} \cdot f(e))^2 = \mathcal{E}(f')\]
    concluding the proof.
\end{proof}

This property of subgraphs allows for a simple algorithm with a linear query complexity that detects if one graph is equal to another, assuming its weights have monotonically increased or decreased. 
}{
First, we show that adding or removing edges from a graph can be tested with $O(n)$ queries.  The proof of the following theorem (that is deferred to the full version) depends on Rayleigh's monotonicity law as well as the following crucial observation: let $v$ be an arbitrary vertex, and let $e$ be an arbitrary edge of $G$.  If the weight of $e$ is changed then the effective resistance of $v$ to at least one vertex of the graph will be changed. Hence, it is sufficient to query ER values between a single vertex $v$ and all other vertices.
}
\begin{theorem}
    \label{thm:proper-subgraph-verification}
    Let $G=(V,E_G,w_G),H=(V,E_H,w_H)$ be $n$-vertex graphs, where $G$ is visible and $H$ is hidden. %
    Suppose $w_G(e) \leq w_H(e)$ for every edge $e \in E_G \cup E_H$ and there exists at least one edge $e'$ such that $w_G(e') < w_H(e')$ strictly (or, similarly, $w_G(e) \geq w_H(e)$ for every $e$ and there exists $e'$ such that $w_G(e') > w_H(e')$). Then, there exists an algorithm that decides whether $G=H$ using $n - 1$ ER queries.
\end{theorem}

\full{
\begin{proof}
    Let $s\in V$ be an arbitrary vertex.
    For all $v\in V\backslash\{s\}$, our algorithm queries $\res_H(s,v)$ and computes $\res_G(s,v)$.  We accept if $\res_H(s,v)=\res_G(s,v)$ for all $v$ and reject otherwise.  If $G = H$ all these ER distances are equal and the algorithm correctly accepts. On the other hand, by \cref{lem:connected-strict-monotonicity}, if $G\neq H$, there will be at least one $v\in V$ for which $\res_G(s,v) \neq \res_H(s,v)$, thus our algorithm correctly rejects.
\end{proof}
}
For unweighted graphs, this implies a linear query complexity algorithm for detecting wether two graphs are equal, assuming one is a subgraph of the other.

\fullornot{
    \subsection{Cut vertices}
}{
    \paragraph{Cut vertices.}
}
\fullornot{
    Now, we use \cref{lem:vertex-cut-iff-triangle-ineqaulity-tight} to test whether a given vertex is a cut vertex, or two given vertices belong to the same biconnected components, both with linear number of ER queries.
    Now, we use \cref{lem:vertex-cut-iff-triangle-ineqaulity-tight} to test whether a given vertex is a cut vertex, or two given vertices belong to the same biconnected components, both with linear number of ER queries.
}{
    Our results for verifying cut vertices and edges are based on the observation that for any three vertices $a, b, c$, the vertex $b$ is  a cut vertex that separates $a$ and $c$ if and only if $\res(a,b) + \res(b,c) = \res(a,c)$.
}

\begin{theorem}
    \label{thm:cut-vertex-verification}
    Let $G=(V,E)$ be an undirected, weighted or unweighted, graph, and let $n=|V|$. 
    \begin{enumerate} [(1)]
        \item There exists an algorithm that decides whether a vertex $v\in V$ is a cut vertex using $2n - 3$ ER queries.
        \item There exists an algorithm that decides whether two vertices $a,b\in V$ are in the same biconnected component using $2n - 3$ ER queries.
    \end{enumerate}
\end{theorem}
\full{
\begin{proof}
    First, we prove (1). To decide if $v$ is a cut vertex, our algorithm fixes some $u \in V \setminus \{ v \}$, and queries $\res(u,v)$.
    For every $w \notin \{u, v\}$, we will then query $\res(u,w)$ and $\res(v,w)$. Therefore, we have a total number of $2(n - 2) + 1 = 2n-3$ queries.
    Our algorithm concludes that $v$ is a cut vertex if and only if there exists a $w\in V$, $w\neq v$, such that $\res(u,w) = \res(u,v) + \res(v,w)$.
    If $v$ is a cut vertex, then removing it separates $u$ from at least one other vertex $w$ in $G$, i.e.~$w$ can be any vertex that is in a different connected component from $u$ in $G\backslash\{v\}$.  In this case, by \cref{lem:vertex-cut-iff-triangle-ineqaulity-tight}, $\res(u,w) = \res(u,v) + \res(v,w)$, and our algorithm correctly decides that $v$ is a cut vertex.  
    On the other hand, if $v$ is not a cut vertex, then $\res(u,w) = \res(u,v) + \res(v,w)$ does not hold for any $w$, by \cref{lem:vertex-cut-iff-triangle-ineqaulity-tight}.  Therefore our algorithm correctly decides that $v$ is not a cut vertex.    

    Next, we prove (2). 
    To decide if $a$ and $b$ belong to the same biconnected compoenent, we query $\res(a,r)$ and $\res(b,r)$ for every $r\in V$, hence $2n-3$ number of ER queries.
    Our algorithm concludes that $a$ and $b$ are in different biconnected component if and only if $\res(a,r) + \res(r, b) = \res(a,b)$ for some $r\in V$, $r\neq a,b$.
    If $a$ and $b$ are in different biconnected components, 
    by \cref{lem:vertex-cut-iff-triangle-ineqaulity-tight},
    there exists an $r$ that makes the triangle inequality tight.  Therefore, our algorithm correctly decides that $a$ and $b$ are in different biconnected components. 
    On the other hand, if $a$ and $b$ belong to the same biconnected component, the triangle inequality is not tight for any $r\in V$, by \cref{lem:vertex-cut-iff-triangle-ineqaulity-tight}.  Hence, our algorithm correctly decides that $a$ and $b$ are in the same biconnected components.
\end{proof}
}

\fullornot{
\subsection{Cut edges}
Not only can we characterize cut vertices by resistance distance, but we can also do the same for cut edge.
For a cut edge, we show that for any vertex in the graph, its resistance distance to each of the endpoints of the edge cut differ by exactly one.
}{
    \paragraph{Cut edges.} For verifying cut edges, we need the following characterization in addition to the tightness of triangle inequality for cut vertices.
}
\begin{lemma}
    \label{lem:cutedge-iff-res-is-plusorminus-one}
    Let $G=(V,E)$ be a connected undirected unweighted graph.
    A pair of vertices $a,b \in V$ is a cut edge if and only if $|\res(a,x) - \res(b,x)| = 1$ for all $x \in V$.
\end{lemma}
\begin{proof}
    Suppose that $|\res(a,x) - \res(b,x)| = 1$ for all $x \in V$.
    Let $A$ and $B$ be sets such that $A = \{ x \in V \ | \ \res(b,x) = \res(a,x) + 1\}$, i.e.~vertices that are closer to $a$ than $b$, and $B = \{ x \in V \ | \ \res(a,x) = \res(b,x) + 1 \}$, i.e.~vertices that are closer to $b$ than $a$.  By the assumption of lemma, $A$ and $B$ partition $V$.  In particular, note $a\in A$ and $b\in B$.
    For any $a' \in A$, we have $\res(a',b)=\res(a',a)+1=\res(a',a)+\res(a,b)$. 
    So, by \fullornot{\cref{lem:vertex-cut-iff-triangle-ineqaulity-tight},}{our characterization of cut vertices (see full version), it follows that} $a$ is a cut vertex that separates any $a'\in A\backslash\{a\}$ from $b$.
    Using the same argument, we conclude that $b$ is a cut vertex that separates any $b' \in B\backslash\{b\}$ from $a$.
    
    Now, let $\pi$ be the shortest $a$ to $b$ path.  Since for every vertex $x$ in $V\backslash\{a,b\}$ either $a$ separates $x$ from $b$ or $b$ separates $x$ from $a$, $\pi$ cannot have any vertex other than $a$ and $b$.  So, $(a,b)\in E$.  Also, $\res(a,b) = |\res(a,b) - \res(b,b)| = 1$. \fullornot{So, by \cref{lem:er_and_edge_cuts}, $(a,b)$ is a cut edge}{So, since $(a,b)$ is an edge with resistance one, it follows that it must be a cut edge, since if $(a,b)$ lies on a cycle it will have an effective resistance less than $1$ (for details, see full version).}

    To see the other direction, we suppose $(a,b)$ is a cut edge.  Let $A\subseteq V$ and $B\subseteq V$ be the subset of vertices in the connected component of $a$ and $b$ in $G\backslash\{(a,b)\}$, respectively.
    \fullornot{We know that $\res(a,b) = 1$ by \cref{lem:er_and_edge_cuts}}{Since $(a,b)$ does not lie on any cycle, it must have an effective resistance of exactly one (see full version)}.  We also know that $a$ is a cut vertex separating $A$ from $b$, and $b$ is a cut vertex separating $a$ from $B$.  Therefore, \fullornot{by \cref{lem:vertex-cut-iff-triangle-ineqaulity-tight}}{since the triangle inequality is tight if and only if it involves a cut vertex}, for any vertex $x\in A$, $\res(x,b) = \res(x,a) + \res(a,b) = \res(x,a) + 1$.
    Similarly, for any vertex $x\in B$, $\res(x,a) = \res(x,b) + \res(a,b) = \res(x,b) + 1$.  Since $\{A, B\}$ is a partition of $V$, $|\res(a,x) - \res(b,x)| = 1$ for all $x\in V$ as desired.
\end{proof}

To decide if a pair of vertices $a,b\in V$ are adjacent via a cut edge, we query $\res(a,x)$ and $\res(b,x)$ for all $x\in V$, and check if $|\res(a,x) - \res(b,x)| = 1$ holds for all of them.  Therefore, we obtain the following algorithm.

\begin{theorem}
    \label{thm:cut-edge-verification}
    Let $G=(V,E)$ be an undirected, weighted or unweighted graph, and let $n=|V|$, and let $a,b\in V$.
    There exists an algorithm that decides whether $(a,b)$ is a cut edge using $2n - 3$ ER queries.
\end{theorem}

\section{Property testing}
Next, we present our algorithms for testing properties of a hidden graph using ER queries. A \emph{graph property} is a set of graphs that is closed under isomorphism, such as connectivity, acyclicity or planarity. The algorithms in this section are approximate in the sense that they guarantee correctness only if the input graph either possesses the desired property or is far from having it. This approximate testing model has been extensively studied within the property testing framework~\cite{books/cu/Goldreich17}. Our results are most closely related to the \emph{bounded-degree model}.

In this model the input graph is assumed to have maximum degree $d$, which is bounded by a constant.  The algorithm can query the neighbors of a vertex.  A graph is $\eps$-far from having a property if more than $\eps d n$ edge modifications are required to make it satisfy that property.  The query complexity in this models is often described as a function of $n$ and $\eps$. See \citet[Section 9]{books/cu/Goldreich17} for more on this model.

\fullornot{
    In this section, we provide property testing algorithms for vertex biconnectivity and edge biconnectivity. Additionally, we present a theorem that allows us to adapt existing property testing algorithms to our query model. The number of queries required depends on the \emph{ER density}, a graph parameter that we introduce and define. We demonstrate that bounded-degree, bounded-treewidth graphs have small ER density, making them well-suited for efficient property testing algorithms under our framework.
}{
    In the full version of the paper, we provide property testing algorithms for vertex biconnectivity and edge biconnectivity. Additionally, we present a  meta theorem that allows for the adaptation of existing property testing algorithms to our query model for bounded treewidth graphs. In the short version, we only present this final result.
}

\full{
\subsection{Testing vertex biconnectivity}
First, we use \cref{thm:cut-vertex-verification} to show that the vertex biconnectivity of a graph can be tested with a linear number of queries. 

\begin{theorem}
\label{thm:vertex_bi_conn_test}
Let
$\eps > 0$ and let $G=(V,E)$ be an undirected graph with $n$ vertices.
There exists an algorithm that makes $O(n/\eps)$ effective resistance queries that accepts with probability $1$ if $G$ is vertex biconnected, and rejects with probability at least $ 2/3$ if $G$ is $\eps$-far from being vertex biconnected, 
i.e., one needs to add %
at least $\eps m$ edges to make $G$ vertex biconnected.
\end{theorem} 
\begin{proof}
First, we check that $G$ is connected by making $n-1$ queries. Specifically, we check that an arbitrary vertex $v$ is at finite distance from all other $n - 1$ vertices.  If $G$ is not connected we reject. So, we assume $G$ is connected in the rest of the proof, which in particular implies that $m \geq n-1$.

Our algorithm takes an arbitrary vertex $v\in V$, and samples $s = 4/\eps$ other vertices $u_1, \ldots, u_s$ uniformly at random. 
It rejects if any of the following conditions holds: (1) $v$ is a cut vertex, 
or (2) there exists $1\leq i\leq s$ such that $v$ and $u_i$ belong to different biconnected components.  If none of these two conditions hold, our algorithm accepts.  Both of these conditions can be checked with $O(s\cdot n) = O(n/\eps)$ ER queries using the algorithm of~\Cref{thm:cut-vertex-verification}.

Clearly, if $G$ is biconnected none of the two conditions holds, thus our algorithm accepts. 

Now, suppose $G$ is $\eps$-far from being biconnected.  If $v$ is a cut vertex our algorithm rejects, so suppose $v$ is not a cut vertex, and it blongs to a biconnected component $B$.  

As $G$ is $\eps$-far from being biconnected, it contains at least 
\(
\eps m + 1 \geq \eps(n-1) + 1
\) biconnected components.
Therefore, its largest biconnected component, and in particular $B$, has at most 
\(
n - \eps(n-1) = n - \eps n + \eps \leq n - \eps n + \eps(n/2) = (1-\eps/2)n
\)
vertices. Therefore, 
with probability at least
\(
1- (1 - \eps/2)^{s} \geq 1 - e^{-(\eps/2)\cdot s} \geq 1 - e^{-2} \geq 2/3,
\)
there exists a $1\leq i\leq s$ such that $u_i\notin B$.  Since, $v$ is not a cut vertex, in this case, $v$ and $u_i$ belong to two different biconnected components, hence our algorithm rejects.
\end{proof}
}

\subsection{Local reconstruction via Schur complement}
\label{subsec:local_reconst_schur_comp}
Our result for adapting known property testing algorithms relies on an algorithm that identifies the neighbors of a given vertex with ER queries. 
\full{
    This algorithm exploits properties of the Schur complement stated in \cref{lem:schur_comp_basics}, as well as the algorithm of \cref{lem:full_reconstruciton}
    for full reconstruction of a graph from its pairwise ER distances.
    \begin{corollary}
    \label{lem:schur_comp_full_sonstruct}
    Let $G = (V, E)$ be a weighted graph, and let $U\subseteq V$.  The Schur complement of $G$ on $U$, $G_U$, can be computed with $\binom{|U|}{2}$ effective resistance queries.
    \end{corollary}
}

We denote the \emph{unit ball} around a vertex $v \in V$ with respect to the effective resistance metric by $B_{\res}(v) := \{u\in V \stbar \res(u,v)\leq 1\}$.
\fullornot{
    It follows from \cref{lem:schur_comp_full_sonstruct} that we can find neighbors of a given vertex, provided that there are not too many vertices in its unit ball in the ER metric of the graph.
}{
    One can reconstruct the Schur complement of $G$ on $B_{\res}(v)$ with $O(|B_{\res}(v)|^2)$ queries.  Since the neighbors of $v$ are all in $B_{\res}(v)$, the set of neighbors of $v$ in $G$ and the set of neighbors of $v$ in this Schur complement are identical.  Hence, we have the following corollary.
}
\begin{corollary}
\label{cor:neighborhood_construct}
Let $G=(V,E)$ be an unweighted graph and $u\in V$.
Then, one can find all neighbors of $u$ with $O(n + |B_{\res}(u)|^2)$ ER queries.
\end{corollary}
\full{
\begin{proof}
We start by computing $B_{\res}(u)$, which we can do by querying for the ER distance between $u$ and each of the other $n - 1$ vertices. 
All neighbors of $u$ in $G$ are at ER distance at most one by \cref{lem:er_and_edge_cuts}, and hence contained in $B_{\res}(u)$.  Therefore, by \cref{lem:schur_comp_basics}, $u$'s neighborhood is preserved in the Schur complement $G_{B_{\res}(u)}$.
So, we can find the neighbors of $u$ by constructing the Schur complement, $G_{B_{\res}(u)}$, which can be done with $O(|B_{\res}(u)|^2)$ ER queries by \cref{lem:schur_comp_full_sonstruct}.
\end{proof}

We will use \cref{cor:neighborhood_construct} in two different ways: to test edge biconnectivity for a general graph, and to adapt property testing algorithms for graphs whose unit balls contain a bounded number of vertices. 
}
\full{
\subsection{Testing edge biconnectivity}
Next, we show an algorithm to test edge biconnectivity of a graph. Our algorithm relies on \cref{cor:neighborhood_construct} to find neighbors of vertices, which facilitates breadth-first search %
in a hidden graph.

\begin{theorem}
\label{thm:edge_bi_conn_test}
Let $G=(V,E)$ be an undirected graph with $n$ vertices and let $0<\eps\leq 1$.
There exists an algorithm that makes $O(n/\varepsilon^2+1/\eps^4)$ ER queries and behaves as follows. It accepts with probability $1$ if $G$ is edge biconnected, and rejects with probability $\geq 2/3$ if $G$ is $\eps$-far from being edge biconnected, i.e., if one needs to add more than $\eps\cdot m$ edges to make $G$ edge biconnected.
\end{theorem}
\begin{proof}
As in the proof of \cref{thm:vertex_bi_conn_test}, we first check that $G$ is connected using $n - 1$ ER queries, and if it is not we reject. 
So, we assume $G$ is connected in the rest of the proof, which in particular implies that $m \geq n-1$.

Our algorithm samples a set $S\subseteq V$ of $16/\eps$ vertices uniformly at random.
We call a biconnected component \emph{small} if its size is at most $4/\eps$.
We call a biconnected component \emph{low-degree} if it is incident to at most two cut edges.  A tree can be obtained by contracting all edge-biconnected components of the graph. Since the average degree of any tree is less than two, it follows that at least half of the biconnected components must have a low degree.

In high level, our algorithm checks if any of the sampled vertices in $S$ belong to a biconnected component that is both small and low-degree with a BFS-like algorithm.  It rejects if it finds such a biconnected component and accepts otherwise.
For each $s\in S$, our algorithm simulates a BFS starting at $s$, that can end in three different ways: (1) the BFS finds a cut edge, hence it decides that $G$ is not edge biconnected, (2) the BFS finds more than $4/\eps + 2$ vertices that are either in the same biconnected component of $s$, or are incident to cut edges of this biconnected component, hence it decides that the biconnected component of $s$ is not small or not low-degree.  

To simulate BFS, upon taking a vertex $u$ from the queue, we query all effective resistances from $u$, and compute the unit ER ball centered at $u$, $B_{\res}(u)$. 
Note, every vertex that is strictly inside this ball is in the same edge biconnected component of $u$.  Also, every vertex on the surface of this ball, is either in the same edge biconnected component, or it is adjacent to $u$ via a cut edge (\cref{lem:er_and_edge_cuts}).  Therefore, if $|B_{\res}(u)|>4/\eps + 2$, the edge biconnected component of $u$ (and also the sampled point that is the root of this BFS call) cannot be both small and low-degree.  In this case, we continue to the next sample point. 
Otherwise, if $|B_{\res}(u)|\leq 4/\eps + 2$, we compute the neighbors of $u$ with another $|B_{\res}(u)|^2 = O(1/\eps^2)$ queries, using the algorithm of \cref{cor:neighborhood_construct}, and check if any of the incident edges of $u$ is a cut edge, by checking if its effective resistance is exactly one, using \cref{lem:er_and_edge_cuts}. 
If we find a cut edge then we reject. Otherwise, we continue the BFS until we visit at least $4/\eps$ vertices. In that case, we continue to the next sampled point, as the biconnected component of the current sampled point is larger than $4/\eps$, and hence not small.

For each vertex $v$ that the BFS visits, it spends $O(n)$ queries to compute $B_{\res}(v)$, and $O(1/\eps^2)$ queries to compute the neighbors of $v$.  For each sampled vertex in $S$, BFS visits $O(1/\eps)$ vertices, the size limit for low-degree small biconnected components.  Therefore, the total number of queries is
\(
|S|\cdot O(1/\varepsilon) \cdot O(n+1/\eps^2) = O(n/\varepsilon^2+1/\eps^4).
\)

If $G$ is edge biconnected, our algorithm accepts with probability one, as it will not find a cut edge.

If $G$ is $\eps$-far from edge biconnectivity, then it contains at least 
\(
\eps m + 1 \geq \eps(n-1)
\) edge biconnected components.  At least half of them, i.e.~$\eps(n-1)/2$ of them, are low-degree.
On the other hand, $G$ (which has $n$ vertices) can contain at most $\eps n/4$ edge biconnected components of size at least $4/\eps$. Thus, assuming $\eps \leq 1$ and $n\geq 8$, there are at least %
\(
\eps(n-1)/2 - \eps n/4 \geq \eps n/8
\)
small edge biconnected components, i.e., of size at most $4/\eps$, that are incident to at most two cut edges.
In particular, these low-degree small biconnected components have at least $\eps n/8$ vertices in total.  Hence, our sample $S$ (of size $16/\eps$) contains at least one vertex from a low-degree small biconnected component with probability at least 
\(
1-(1-\eps/8)^{-16/\eps} \geq 1-e^{-2}\geq 2/3.
\)
Therefore, our algorithm rejects with at least this probability.
\end{proof}
}
\subsection{The ER density and property testing with ER queries} \cref{cor:neighborhood_construct} implies the following meta theorem, which allows us to adapt any property testing algorithm within the
bounded-degree model to a property testing algorithm with ER queries.  
The query complexity of our algorithm depends on the ER density $\rho_{\res}(G)$ of the graph $G$, which we define as the maximum size of any ER unit ball, i.e., $\rho_{\res}(G) := \max_{v\in V}{|B_{\res}(v)|}$.

\begin{theorem}
\label{thm:adj_list_pt_to_er_pt}
    Let $G=(V,E)$ be a graph with ER density $\rho = \rho_{\res}(G)$.  Let $A$ be any property testing algorithm for problem $P$ in the bounded-degree model with $f(\varepsilon, n)$ running time.  Then, there is a property testing algorithm for problem $P$ that uses $O(f(\varepsilon, n)\cdot (n+\rho^2))$ effective resistance queries.
\end{theorem}

\begin{remark}
We remark that the linear dependency on $n$ in 
\full{\cref{thm:edge_bi_conn_test} and} \cref{thm:adj_list_pt_to_er_pt} can be eliminated if we have an oracle that returns the vertices in the order of their distance to a vertex $v$, allowing us to stop querying at any point.   Then, to discover the vertices in a unit ball around $v$, $B_{\res}(v)$, we only need to spend $|B_{\res}(v)|$ queries as opposed to $n-1$.
\end{remark}

Property testing in the bounded-degree model is a well-studied subject with many fundamental results.  To name a few, there are $f(1/\eps)$ time testers to decide if a graph (1) contains a fixed subgraph, (2) is regular, (3) is $k$-connected, or (4) belongs to a minor closed family (e.g.,~planar graphs). (See~\citet{books/cu/Goldreich17} for a more extensive list of results.)

\Cref{thm:adj_list_pt_to_er_pt} can be used to adapt all these results to algorithms with ER queries. However, this comes with an additional cost of a factor of $(n + \rho^2)$. There exist graphs for which $\rho^2$ is large, making the theorem less useful—for example, in the case of expander graphs.

However, we show that $\rho$ is bounded for bounded-degree bounded treewidth graphs. As a result, all the aforementioned results from the bounded-degree model can be adapted to bounded treewidth graphs with an additional factor of $n$ in the number of queries. In most cases, this leads to a linear number of ER queries, as opposed to the quadratic number required when reconstructing the entire graph.
\paragraph{ER density of bounded treewidth graphs.}
We use the following two fundamental results of \citet{journals/dmtcs/Bodlaender99} on domino tree decompostion and \citet{journals/mpcps/Nash-Williams59} on a lower bound for ER distance of two vertices towards showing that bounded-degree bounded treewidth graphs have boudned ER density.
\begin{lemma}[{\citet[Theorem 3.1]{journals/dmtcs/Bodlaender99}}]
\label{lem:good_tree_decomp}
    Let $G=(V,E)$ be a graph with maximim degree $d_G$ and treewidth $w_G$.  Then, $G$ has a tree decomposition of width $O(w_G d_G^2)$, in which each vertex is in at most two bags.%
    \footnote{Such a decomposition is called a \emph{domino decomposition}, and the smallest width of a domino decomposition is called \emph{domino treewidth}.} Moreover, the degree of each bag in the tree decomposition is $O(w_G d_G^2)$.
\end{lemma}
    
The last part of the statement is not explicit in \citet{journals/dmtcs/Bodlaender99}.  However, it is easy to conclude by observing that each neighbor of a bag $B$ has at least one vertex in common with $B$, assuming the underlying graph is connected.  Therefore, since each vertex appears in at most two bags, the number of neighbors of a bag cannot be larger than its size.
\begin{lemma}[\citet{journals/mpcps/Nash-Williams59}]
\label{lem:nash_williams}
    Let $G=(V, E)$ be a graph and let $s, t\in V$.
    Let $S_1, S_2, \ldots, S_k$ be such that $s\in S_i$ and $t\not\in S_i$ for all $i\in[k]$, and such that the edge sets $E(S_i, V\backslash S_i)$ are pairwise disjoint. Then,\full{
    \[
    \res(s,t) \geq \sum_{i=1}^k{\frac{1}{|E(S_i, V\backslash S_i)|}} \ \text{.}
    \]}{ \(
    \res(s,t) \geq \sum_{i=1}^k|E(S_i, V\backslash S_i)|^{-1}\).}
\end{lemma}

\begin{lemma}
\label{lem:er_vs_tree_decomp_distance}
    Let $G=(V,E)$ be a graph with maximum degree $d_G$, and let $({\cal B}, T)$ be a tree decomposition of $G$ with width $w_T$ and maximum degree $d_T$.
    Also, suppose that each vertex of $G$ appears in at most $b_T$ bags of $T$.
    Let $s,t\in V$, let $B_s, B_t\in {\cal B}$ be any two bags that contain $s$ and $t$, respectively, and let $r_T(s,t)$ be the distance of these bags in $T$. 
    \[
    r_T(s,t) \leq {4b_T\cdot d_G\cdot w_T}\cdot \res(s,t)  \ \text{.}
    \]
\end{lemma}
\begin{proof}
    If $r_T(s,t) < 4b_T$, the statement of the lemma holds. Thus, assume $r_T(s,t) \geq 4b_T$.
    
    Let $p = r_T(s,t)$, and let $B_s = B_1, \ldots, B_p = B_t$ be the unique $B_s, B_t$ path in $T$.
     We know that each vertex of $G$ is in at most $b_T$ bags.  Moreover, we know that these bags form a connected induced tree in the tree decomposition.
     Therefore, for each 
     $i\in [p-b_T]$, $B_i$ and $B_{i+b_T+1}$ are disjoint.  
     We conclude that there is a subsequence $B'_1,\ldots, B'_q$ of $B_1, \ldots, B_p$ such that (1) $q \geq \lfloor p/b_T\rfloor  - 2$, (2) $s,t\not\in B'_i$ for $i\in [q]$, and (3) for each $1\leq i< j\leq q$, $B'_i$ and $B'_j$ are disjoint (we need $-2$ in condition (1) because of condition (2)).  Note, $q\geq 1$ because $p\geq 4b_T$. 
     
     Note, for each $i\in [q]$, $B'_i$ is a vertex cut that separates $s$ and $t$ in $G$. Let $S^i_1, \ldots, S^i_{r_i}$ be the connected components of %
     $V \setminus B_i'$ with $s\in S^i_1$ and $t\in S^i_{r_i}$.  Then, let
     $S'_i := S^i_1\cup \cdots \cup S^i_{r_i -1}\cup B'_i$, and observe that $s\in S'_i$ and $t\not\in S'_i$.

     Moreover, for each $i\in [q]$, each edge in %
     $E(S'_i, V\backslash S'_i)$%
     has one endpoint in $B'_i$ (note, in particular, $V \setminus S_i' = S_{r_i}^i$).  Since the bags $B'_i$ are disjoint, we conclude that the edge sets $E(S'_i, V\backslash S'_i)$ for $i\in [q]$ are disjoint.  We also conclude that for every $i \in [q]$,
     \[
        |E(S'_i, V\backslash S'_i)| \leq d_G\cdot |B'_i| \leq d_G\cdot w_T.
    \]
     It follows by \cref{lem:nash_williams} that %
     \begin{align*}
     \res(s,t) \geq& \frac{q}{d_G\cdot w_T} \geq \frac{\lfloor p/b_T\rfloor  - 2}{d_G\cdot w_T}
     \geq \frac{p/b_T  - 3}{d_G\cdot w_T} \\
     \geq& \frac{p/b_T  - (3/4)\cdot (p/b_T)}{d_G\cdot w_T} 
      = \frac{(1/4)\cdot p}{b_T\cdot d_G\cdot w_T} =  \frac{(1/4)\cdot r_T(s, t)}{b_T\cdot d_G\cdot w_T},
     \end{align*}
     where the last inequality holds as $p\geq 4b_T \Rightarrow 3 \leq (3/4)\cdot(p/b_T)$.
\end{proof}
\begin{theorem}
\label{thm:er_density_tw}
    Let $G=(V,E)$ be a graph with maximum degree $d_G$ and treewidth $w_G$.
    Then the ER density of $G$ is $(w_Gd_G)^{O(w_G d_G^3)}$.  In particular, a graph with constant maximum degree and constant treewidth has constant ER density.
\end{theorem}
\begin{proof}
    By \cref{lem:good_tree_decomp} there exists a tree decomposition $T$ of $G$ with maximum degree $d_T=O(w_Gd_G^2)$ and width $w_T=O(w_Gd_G^2)$ in which each vertex appears in at most $b_T = 2$ bags.

    Let $s\in V$, and let $B_s$ b ethe bag in $T$ that contains $s$.  By \cref{lem:er_vs_tree_decomp_distance}, vertices with ER distance at most one from $s$ (i.e., vertices in $B_{\res}(s)$), are in bags of distance at most $\alpha := \max(4b_T\cdot d_G\cdot w_T, 4 b_T) = 4b_T\cdot d_G\cdot w_T$ from $B_s$ in $T$.

    There are at most %
    $
    \sum_{i=0}^{\alpha} d_T^i = {(d_T^{\alpha+1}-1)}/{(d_T - 1)} \leq d_T^{\alpha+1}
    $ 
    such bags, assuming $d_T \geq 2$, that collectively contain at most
    \[
    w_T\cdot d_T^\alpha = 
    w_T\cdot d_T^{4b_T\cdot d_G\cdot w_T + 1} = 
    w_Gd_G^2\cdot (w_Gd_G^2)^{O(d^3_G\cdot w_G)} = 
    (w_Gd_G)^{O(d^3_G\cdot w_G)}
    \]
    vertices.
\end{proof}
Therefore, the ER density of a graph $G$ is a constant if it has constant maximum degree and constant treewidth.  In that case, \cref{thm:adj_list_pt_to_er_pt} implies a property testing algorithm with a linear number of ER queries, if a constant-time property testing algorithm is known in the bounded-degree model.

\section{Shortest Path versus Effective Resistance Queries}
\label{sec:sp-versus-er}

In this section, we motivate the study of effective resistance queries by showing that they are incomparable in power to shortest path queries. Specifically, we show that sometimes it is (provably) possible to use fewer effective resistance queries to solve a given task, and sometimes it is possible to use fewer shortest path queries.

Our first result asserts that it is more efficient to check whether a graph is a clique using effective resistance queries than shortest path queries.
\begin{theorem} \label{thm:er-sp-clique}
There is an algorithm for testing whether a graph $G$ with $n$ vertices is a clique (i.e., whether $G = K_n$) using $O(n)$ effective resistance queries, but any such algorithm requires $\Omega(n^2)$ shortest path queries. 
\end{theorem}

\begin{proof}
As shown in \cref{thm:proper-subgraph-verification}, it is possible to check using $O(n)$ effective resistance queries whether two graphs $G$ and $H$ are equal, assuming that $G$ is a subgraph of $H$. The algorithm for checking that $G$ is a clique follows from the special case of this result when $H = K_n$, as every graph is a subgraph of $K_n$.

The lower bound for shortest path queries works as follows. Suppose that $G$ is either equal to $K_n$ or equal to $K_n \setminus \{e\}$ for a single unknown edge $e$. Any shortest path query on $u, v$ will return the same value (i.e., $1$) in either of these two cases, except when $u$ and $v$ are the endpoints of $e$. It follows that any algorithm requires $\binom{n}{2}$ shortest path queries.
\end{proof}

Our second result asserts that it is more efficient to check whether two given vertices in $G$ are adjacent using shortest path queries than effective resistance queries.
\full{
\begin{figure}
    \begin{center}
    \setlength\tabcolsep{2cm}.
    \begin{tabular}{cc}
    \includegraphics[scale=0.55]{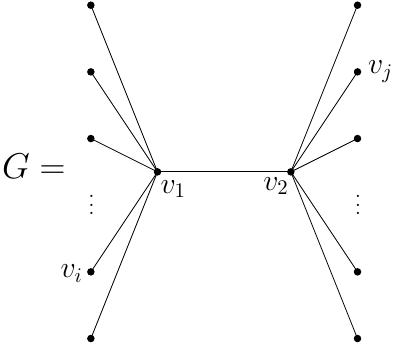} &
    \includegraphics[scale=0.55]{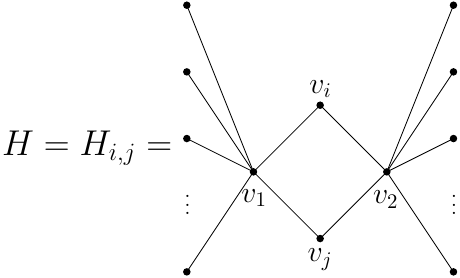}
    \end{tabular}
    \end{center}
    \caption{\small A pair of unweighted graphs $G$ and $H_{i,j}$ on $n$ vertices that require $\Omega(n)$ effective resistance queries to distinguish; see \cref{thm:er-sp-adjacency}.  The effective resistance between $v_1, v_2$ is the same in both graphs (i.e., $R_G(v_1, v_2) = R_H(v_1, v_2) = 1$) despite $v_1, v_2$ being adjacent in $G$ but not in $H_{i,j}$.}
    \label{fig:effect-res-inefficient}
\end{figure}
}
\begin{theorem} \label{thm:er-sp-adjacency}
There is an algorithm for testing whether two given vertices $u, v \in V$ in a simple graph are adjacent using $1$ shortest path query, but any algorithm for this problem requires $\Omega(n)$ effective resistance queries.
\end{theorem}
\full{
\begin{proof}
It is possible to check whether $u, v$ are adjacent using a single shortest path query since $u$ and $v$ are adjacent if and only if the length of the shortest path between them is $1$.

On the other hand, consider the problem of deciding whether the underlying graph is (1) $G$ or (2) of the form $H = H_{i,j}$ for some $i, j$ using effective resistance queries, where these graphs are as shown in \cref{fig:effect-res-inefficient}.
I.e., $G$ consists of two star graphs on $n/2$ vertices whose centers $v_1, v_2$ are connected by an edge, and $H = H_{i,j}$ consists of two star graphs on $n/2 - 1$ vertices whose centers $v_1, v_2$ form a square with the vertices $v_i, v_j$, where $i, j \in \set{3, \ldots, n}$, $i \neq j$ are unknown.
In particular, the graphs $G$ and $H$ are on the same set of vertices $V = \set{v_1, \ldots, v_n}$, but in $G$, $v_1$ and $v_2$ are adjacent, and in $H$, $v_1$ and $v_2$ are not adjacent.
So, the problem of deciding whether the underlying graph is $G$ or $H$ reduces to checking whether $v_1$ and $v_2$ are adjacent.
To show that checking adjacency requires $\Omega(n)$ effective resistance queries, it therefore suffices to show that determining whether the underlying graph is $G$ or $H$ requires $\Omega(n)$ such queries.

One can check that $\res_G(x, y) = \res_H(x, y)$ for all $x, y \in V \setminus \set{v_i, v_j}$, and in particular that $\res_G(v_1, v_2) = \res_H(v_1, v_2) = 1$. 
So, to decide whether the underlying graph is $G$ or $H$, it is necessary to perform a query involving either $v_i$ or $v_j$. However, $i, j$ are unknown, and so any (possibly adaptive) sequence of queries $\res(x, y)$ necessarily including $v_i$ and $v_j$ must be of length at least $(n - 2)/2 = \Omega(n)$.
\end{proof}
}

Put together, \cref{thm:er-sp-clique,thm:er-sp-adjacency} show that shortest path queries and effective resistance queries can each be asymptotically more efficient than the other for some verification tasks, and so neither type of query subsumes the other---they are incomparable.
Understanding their relative power is therefore quite interesting.
We also note that \cref{thm:er-sp-adjacency} shows that testing whether a given pair of vertices is adjacent takes $\Omega(n)$ effective resistance queries, but we do not know if this is tight. In fact, we do not have a better upper bound than the trivial $\binom{n}{2}$ queries needed to recover the whole graph. Finding tight bounds for this problem is therefore an interesting question.

\section{Graph Reconstruction via Effective Resistance Queries}
\label{sec:reconstruction}
We give algorithms solving two graph inference problems, the first being full reconstruction of a graph from resistance distance queries given a tree decomposition, and the second being the problem of completing a graph for which the adjacency matrix is only partially known.
For the former, 
\fullornot{}{which is deferred to the full version},
we utilize the technique of reconstructing the Schur complement for a subset of vertices to obtain the overall graph. For the latter one, we rely on the property that effective resistance appears as directional derivatives of certain functions of the space of positive definite matrices.

\full{
\subsection{Graph Reconstruction from a Tree Decomposition}
Next, we show that a graph can be reconstructed quickly provides a tree decomposition of a small width.  Estimating such a tree decomposition is a more challenging task that we leave as an open question.
Our results in this section relies on properties of Schur complement of \cref{lem:schur_comp_basics}.

\begin{corollary}
    \label{lem:schur-bag-interior-reconstruct}
    Let $G = (V, E)$ be a graph, and let $({\cal B}, T)$ be a tree decomposition of $G$ of width $k \in \mathbb{N}$.  Also, let $B \subseteq {\cal B}$ be a bag of $T$ that has degree one, and let $P \subseteq {\cal B}$ be the parent of $B$.
    For all vertices $u \in B \setminus P$ we can determine their neighborhood using $k \choose 2$ effective resistance queries.
\end{corollary}
\begin{proof}
    Any vertex $u \in B \setminus P$ does not appear in any other bag of $T$, because vertices only appear in connected induced trees of $T$.  Thus, all its neighbors have to appear in $B$, because every edge of the graph must appear in at least one bag.
    It follows, by \cref{lem:schur_comp_basics} that the neighborhood of $u$ is identical in $G$ and the Schur complement $G_{B\backslash P}$.  Thus, we can discover them
    by reconstructing $G_{B\backslash P}$, with $\binom{|B|}{2} = \binom{k}{2}$ ER queries using \cref{lem:full_reconstruciton}.
\end{proof}

Using this, we can obtain the adjacency of vertices that belong only in a leaf bag.
Recursively, the vertices exclusive to that leaf bag can be eliminated via the Schur complement.
What we show is that removing this bag, as a node, from the tree decomposition, gives a valid tree decomposition for the new graph, with those vertices being eliminated.
\begin{lemma}
    \label{lem:schur-complement-bag-removal}
    Let $G=(V,E)$ be a graph with a tree decomposition $({\cal B},T)$ of treewidth $k$.
    For any bag of degree one, $B \in {\cal B}$, and its parent $P \in {\cal B}$, the subtree $T'$ induced by ${\cal B}' = {\cal B} \setminus \{B\}$ forms a valid tree decomposition of the Schur complement graph $G' = G_{V\backslash(B \setminus P)}$.
\end{lemma}
\begin{proof}
    We must show (1) every vertex of $V\backslash(B\backslash P)$ is present in a bag of ${\cal B}'$, (2) for any edge in $G'$ there is a bag in ${\cal B}'$ that contains both of its endpoints, and (3) every vertex of of $G'$ exists in a contiguous subtree of $T'$.
    Property (3) is guaranteed since we removed a leaf from $T$ to obtain $T'$, which does not disconnect any induced subtree of $T$.

    To show property (1), note the vertex set of $G'$ is exactly $V \setminus (B \setminus P)$.
    Let $v \in V \setminus (B \setminus P)$.
    If $v \notin B$ then there is a bag $B'\neq B$ such that $v \in B'$, otherwise $({\cal B}, T)$ is not a valid tree decomposition.
    Otherwise, if $v \in B$ then it must be the case that $v\in P$, since $v\in V \setminus (B \setminus P)$.  Therefore, $v$ is in a bag of ${\cal B}'$, namely $P$.

    Finally, we show property (3). For any edge $e$ in the Schur complement graph $G_{V\backslash(B \setminus P)}$, either $e\in E$ or $e$ is created during the Schur complement process. 
    If $e\in E$ then it is in a bag of ${\cal B}$, therefore, in a bag of ${\cal B}\backslash\{B\}$. 
    Otherwise, $e = (u,v)$ such that $u,v$ are in the neighborhood of $B\backslash P$, as the edge is created after reducing $B\backslash P$.
    In this case, since all the neighbors of $B\backslash P$ are in $P$, we have that $u,v\in P$, as desired.
\end{proof}
Having shown both that we can obtain the entire adjacency of `interior vertices' of a leaf bag, and that removing such bags gives a tree decomposition of the Schur complement with those vertices eliminated, the algorithm becomes clear.
\begin{theorem}
    \label{thm:reconstruction-from-tree-decomp}
    Let $G=(V,E)$ be an undirected weighted graph and suppose that a tree decomposition $({\cal B}, T)$ of $G$ is given that has width $k$.
    There exists an algorithm to reconstruct $G$ in $O(k^2n)$ ER queries.
\end{theorem}
\begin{proof}
    Our algorithm inductively reconstructs $G$ by considering its tree decomposition $T$.
    Let $B$ be a bag of $T$ with exactly one neighbor $P$ in $T$, and let $U = B\backslash P$.  We can assume $U$ is non-empty, as otherwise, we can drop $B$ and obtain a tree decomposition of $T$ with fewer bags.  We can compute all the neighbors of all the vertices in $U$ with $O(k^2)$ ER queries by \cref{lem:schur-bag-interior-reconstruct}.  Then, we recursively compute the Schur complement $G_{V\backslash U}$.  To that end, we use the fact that $T\backslash\{B\}$ is a valid tree decomposition of $G_{V\backslash U}$, which holds by \cref{lem:schur-complement-bag-removal}, and that the ER queries between two vertices of $V\backslash U$ return their effective resistance in this Schur complement (which is equal to their effective resistance in $G$).

    To finish our reconstruction, we need to build $G$ provided the neighborhood of $U$ (computed by the algorithm of \cref{lem:schur-bag-interior-reconstruct}) and the Schur complement $G_{V\backslash U}$ (computed recursively). To that end, let $L$ be the Laplacian of $G$, and observe that we know $L(U, \overline{U})$, $L(\overline{U}, U)$ and $L(U, U)$, since we know the neighbors of every vertex in $U$.  Moreover, we know $L_{\overline{U}}$ from our recursive computation. Hence, we can compute $L(\overline{U}, \overline{U})$, as by the definition of Schur complement (\cref{eqn:schur_comp}), 
    \[
    L(\overline{U}, \overline{U}) = L_{\overline{U}} + L(\overline{U},U)L(U, U)^{+}L(U, \overline{U}).
    \]
    Overall, we know all four submatrices of $L$, hence we know $L$.
\end{proof}
}

\fullornot{
\subsection{Graph Completion with Minimal Query Complexity}
Next, we consider the problem of graph completion. 
}{
\paragraph{Graph Completion.}
}
Let $G = (V, E)$ be a connected weighted graph whose adjacency matrix is partially known, with exactly $k$ unknown entries. The goal of the graph completion problem is to identify these unknown entries using the minimum number of ER queries.

\fullornot{It is straightforward to recover $G$ using $O(n^2)$ ER queries by simply ignoring the known part of the adjacency matrix and applying \fullornot{\cref{lem:full_reconstruciton}}{}. Furthermore, by leveraging the Schur complement, one can recover $G$ with $O(k^2)$ ER queries (\fullornot{\cref{thm:k-unknowns-quadratic_queries}}{See full version}).
However, we show that only $O(k)$ ER queries are sufficient to reconstruct $G$ when the edge weights are drawn from a finite set, and in particular, for unweighted graphs (\cref{thm:k-unknowns-exponential-completion}).
}{We show that only $O(k)$ ER queries are sufficient to reconstruct $G$ when the edge weights are drawn from a finite set, and in particular, for unweighted graphs (\cref{thm:k-unknowns-exponential-completion}).}
\full{
\subsubsection{Quadratically many ER queries}  To warm up, we present the following simple application of the Schur complement together with the known algorithm for fully reconstructing a graph, to recover $G$ with $O(k^2)$ ER queries.

\begin{theorem}
    \label{thm:k-unknowns-quadratic_queries}
    Let $G=(V,E, w_G)$ be an weighted graph.
    Suppose $k$ entries of the adjacency matrix of $G$ are not known. 
    There exists an algorithm to reconstruct $G$ (i.e.~to find the missing $k$ weights) with $O(k^2)$ ER queries.  (The running time of the algorithm is $\poly(k)\cdot n^2$).
\end{theorem}
\begin{proof}
Let $L$ be the Laplacian of $G$.  Let $U$ be the set of vertices that are incident to the edges with unknown weights, and note $|U|\leq 2k$.  
We use $\binom{2k}{2}$ ER queries to obtain all ER values between all pairs of vertices of $U$, and use \cref{lem:full_reconstruciton} to compute the Schur complement $L_U$.  Then, since $L(U, \overline{U})$, $L(\overline{U}, U)$ and $L(\overline{U}, \overline{U})$ are already known, we can compute $L(U,U)$ from the Schur complement formula, \cref{eqn:schur_comp},
\(
L_U = L(U, U) - L(U,\overline{U})L(\overline{U}, \overline{U})^{+}L(\overline{U}, U) \ \text{.}
\) 
Therefore, we can have all four submatrices of $L$, hence $G$ is recovered.
\end{proof}
}
\fullornot{
\subsubsection{Linearly many ER queries}
Next, }{Specifically, }we show that the effective resistance values between every pair of vertices whose corresponding adjacency matrix entry is missing are sufficient to reconstruct $G$. Specifically, we show that these ER values, together with the known adjacency values, uniquely determine the graph.

The proof of this uniqueness result is, at a high level, very similar to the idea presented in the MathOverflow post by \citet{Math-Overflow-Speyer}. However, we needed to adapt that idea to our setting with effective resistance (\cref{thm:unique-graph-completion-k-unknowns}) and establish a connection between the derivative of the function studied there and effective resistance (\cref{lem:derivitive_in_edge_lap_direction}). %
The application of this work to the context of effective resistance is not immediately obvious, as, firstly, the Laplacian is not an invertible matrix, and secondly, the construction of a function with strict convexity whose gradient is equal to a certain set of effective resistance is non-trivial.

\full{
\paragraph{Basic matrix calculus.}
We briefly introduce some background from matrix calculus that we use in this section and refer the reader to~\citet{Magnus_Neudecker_2019} for a comprehensive review.

We use $X=(x_{i,j})_{(i,j)\in[n]\times [n]}$ to denote $n\times n$ matrices of variables $x_{i,j}$.
We work with functions $f:\R^{n\times n}\to \R$ that map $n\times n$ matrices to real numbers.
For such an $f$ its \emph{derivative} is an $n\times n$ matrix, denoted %
\[
    \partial f/\partial X = (\partial f/\partial x_{i,j})_{(i,j)\in[n]\times [n]} \text{.}
\]
Moreover, the \emph{directional derivative} of $f$ in the direction of an $n\times n$ matrix $M$ is %

\begin{equation}
\label{eqn:dir_derivative}
\nabla_M f = \text{tr}((\partial f/\partial X)^T\cdot M).    
\end{equation}

A differentiable function $f:D\to\R$ for $D\subseteq \R^{m}$ is strictly concave if $D$ is convex and for any $x, y\in D$
\[
    f(x) - f(y) < (\nabla f(x))^T (y - x).
\]
See~\citet[Section 3.1.3]{book/boyd-convex-opt}.
One can show that strict concavity implies uniqueness of gradients over the domain. Hence, we have the following lemma. %
\begin{lemma}[Theorem 1 of~\citep{Griewank1991}]
    \label{fact:cvx-fn-cvx-dom-implies-different-gradient}
    Let $f:D\rightarrow \R$ be a strictly concave function on a convex domain $D\subseteq \R^{m}$.  Then
    \(
    \nabla f(x) \neq \nabla f(y)
    \)
    for any distinct $x,y\in D$.
\end{lemma}

We consider the $\log\det(\cdot)$ function defined on the set of $n\times n$ matrices.  We specifically work on the restriction of this function on the set positive definite matrices, denoted $S_n^{++} \subseteq \mathbb{R}^{n \times n}$ %
, which are invertible since all their eigenvalues are non-zero. Furthermore, the set $S_n^{++}$ is known to form a cone and, in particular, is convex. 
We use the following fact, whose proof can be found in~\citet[Section 3.5.1]{book/boyd-convex-opt}.
\begin{lemma}
    \label{lem:log-det-concave}
    The function $f : S_n^{++} \to \mathbb{R}$, defined as $f(X) = \log\det (X)$, 
    is strictly concave.
\end{lemma}
}

\paragraph{Uniqueness and reconstruction.}  
Now, we are ready to present the main result. First, we establish the uniqueness of graph completion given the ER values corresponding to the pairs of vertices with unknown entries in the adjacency matrix. %

To this end, we study the function $f = \log\det(\cdot)$ and its derivative on the space of positive definite matrices, which includes regularized Laplacians. First, we show that the directional derivatives of $f$ in certain directions correspond to the effective resistances in graphs.  In the following proof, for any $i,j\in [n]$ $E_{i,j}$ is a matrix whose entry at $(i,j)$ is one, and all its other entries are zero.  Further, $L_{i,j} = E_{i,i}+E_{j,j}-E_{i,j}-E_{j,i}$.  We refer to $L_{i,j}$ as the edge Laplacian, as it is the Laplacian of a graph with $n$ vertices and one edge $(i,j)$.

\begin{lemma}
    \label{lem:derivitive_in_edge_lap_direction}
    Let $G=([n], E_G, w_G)$ be a weighted connected graph and $\widetilde{L}_G$ its regularized Laplacian.  
    For any $i,j\in [n]$, we have,
    \(
        \nabla_{L_{i,j}} \log(\det(X))|_{X = \widetilde{L}_G} = \res(i,j),
    \)
    where $\log\det(\cdot):S_n^{++}\rightarrow \R$ is a function on all positive definite $n\times n$ matrices $X$.
\end{lemma}
\fullornot{
\begin{proof}
    We know $\partial \log( \det(X))/\partial X = (X^{-1})^T$ (see~\citet[Equation (51)]{book/matrix_cookbook}).  
    So, \full{by \cref{eqn:dir_derivative},} 
    \[
        \nabla_{L_{i,j}} \log( \det(X)) = \text{tr}((X^{-1})^T\cdot L_{i,j}) = ((X^{-1})^T)_{i,i} + ((X^{-1})^T)_{j,j} - ((X^{-1})^T)_{i,j} - ((X^{-1})^T)_{j,i}
    \]     
    Evaluated at the particular point $X = \widetilde{L}_G$, and noting that $\widetilde{L}^{-1}_G = (\widetilde{L}^{-1}_G)^T$ , we obtain 
    \[
        \nabla_{L_{i,j}} \log(\det(X))\Big |_{X = \widetilde{L}_G} =  \widetilde{L}^{-1}_{ii} + \widetilde{L}^{-1}_{jj} - 2 \widetilde{L}^{-1}_{ij} = \res(i,j), 
    \]
    where we used $\widetilde{L} = \widetilde{L}_G$ to simplify the notation.
\end{proof}
}
{
This follows clearly from the definition of the directional derivative, which is defined, in the direction of $M$, to be $\nabla_Mf:=\text{tr}((\partial f/\partial X)^TM)$, and the fact that $\partial \log \det(X) / \partial X = (X^{-1})^T$. Using the above lemma, we show that the resistances on the unknown vertex pairs of $G$, are unique for each completion of the graph.
}
\full{Next, we present our uniqueness result based on the lemma above. At a high level, we consider the space $W$ of all regularized Laplacians that are consistent with the known part of $G$. We show that the function $\log\det(\cdot)$ has distinct derivatives at every pair of points in $W$. Moreover, these derivatives are uniquely determined by the effective resistances between pairs of vertices with unknown adjacency entries.  We conclude that knowing ER values between pairs of vertices with missing adjacency values uniquely determines $G$.}
{}
\begin{theorem}
    \label{thm:unique-graph-completion-k-unknowns}
    Let $n\in\Z^+$, and let $I\sqcup J$ be a partition of $\binom{[n]}{2}$, i.e.,~all (unordered) pairs of indices of $[n]$.
    Also, let $w:I\rightarrow \R^{\geq 0}$ and $r:J\rightarrow \R^+$.
    There exists at most one %
    weighted graph $G=([n], E, w_G)$ whose edge weights on $I$ are consistent with $w$, and whose effective resistance values on $J$ are consistent with $r$, i.e.~for any $(i,j)\in I$, $w_G(i,j) = w(i, j)$ and for any $i,j\in J$, $\res_G(i,j) = r(i,j)$.
\end{theorem}
\begin{proof}
    Let $K=([n], I, w)$ be the graph that has edges between pairs in $I$ with their corresponding weights from $w$, and let $\widetilde{L}_K$ be the regularized Laplacian of $K$.
    
    We know, for any graph $G=([n], E_G, w_G)$, that
    \[
    \widetilde{L}_G = \sum_{(i,j)\in E_G}{w_G(i,j)\cdot L_{i,j}} + \frac{1}{n}\cdot J. 
    \]
    Further, if $G$ is consistent with $I$, $w$, we know that $I\subseteq E_G$, and that for any $i,j\in I$, $w_G(i,j) = w(i,j)$.  Therefore,
    \[
    \widetilde{L}_G = \sum_{(i,j)\in I}{w(i,j)\cdot L_{i,j}} + \sum_{(i,j)\in E_G\backslash I}{w_G(i,j)\cdot L_{i,j}} + \frac{1}{n}\cdot J = \widetilde{L}_K + \sum_{(i,j)\in E_G\backslash I}{w_G(i,j)\cdot L_{i,j}}. 
    \]
    Therefore, $\widetilde{L}_G$ belongs to the following affine subspace of $\R^{n\times n}$:
    \[
        W = \widetilde{L}_K + \text{span} \{ L_{i,j} \ | \ (i,j) \in J \}.
    \]
    Moreover, $\widetilde{L}_G$ belongs to the cone of positive definite matrices, $S_n^{++}$, because it is the regularized Laplacian of a connected graph.

    Now, let $f:\R^{n\times n}\to\R$, be $f(X) = \log\det(X)$. Let $G$ be any graph with regularized Laplacian $\widetilde{L}_G\in W$, and let $i,j\in I$.
    By \cref{lem:derivitive_in_edge_lap_direction}, we have
    \[
        \nabla_{L_{i,j}} f(\widetilde{L}_G) = \nabla_{L_{i,j}} f_{|W}(\widetilde{L}_G) = \res_{\widetilde{L}_G}{(i,j)},
    \]
    where the first equality holds as $\widetilde{L}_G + L_{i,j}$ is within $W$.

    By \fullornot{\cref{lem:log-det-concave}, $f$ is strictly concave within $S_n^{++}$}{the strict concavity of the log-determinant function~\citep[Chapter 3.1]{book/boyd-convex-opt}}, it is strictly concave %
    within $W \cap S_n^{++}$ \fullornot{(any line through $W \cap S_n^{++}$ is also a line through $S_n^{++}$, so the concavity is carried over into the intersection.)}{(an intersection of two convex sets).}

    Let $g:\mathbb{R}_{\geq 0}^k \to \mathbb{R}^{n \times n}$ be the mapping from edge weights, for pairs $(i,j) \in J$, to a regularized graph Laplacian given by $g(\vec{w}) = \sum_{(i,j) \in J}w_{i,j} L_{i,j} + \tilde{L}_K$ where $\vec{w} = (w_{i,j})_{(i,j) \in J}$.
    Now define $h(\vec{x}):=(f \circ g)(\vec{x})$, so that $\partial h / \partial \vec{x}$ is equal to the vector of all directional derivatives, $(\nabla_{L_{i,j}}f(\tilde{L}_K))_{(i,j) \in J}$.
    Its strict concavity is given by $f$ being strictly concave and $g$ being linear.
    To see that $\partial h / \partial x_{i,j} = \nabla_{L_{i,j}}f(\tilde{L}_{K})$, note that $h(\vec{0} + \varepsilon \vec{1}_{i,j}) = f(\tilde{L}_K + \varepsilon L_{i,j})$.
    Then for any two completions of the graph, $G$ and $H$, let $\vec{w}_G$ and $\vec{w}_H$ be the corresponding edge weights given by those completions of the graph.
    By \fullornot{\cref{fact:cvx-fn-cvx-dom-implies-different-gradient},}{the injectivity of the gradient of a strictly convex or concave function~\citep{Griewank1991},} $\partial h(\vec{w}_G) /\partial \vec{x}= \partial h(\vec{w}_H) / \partial \vec{x}$ implies that $\vec{w}_G = \vec{w}_H$, which means that $G$ is equal to $H$.
    Hence the completion of the graph corresponds to a unique set of effective resistances on indices in $J$.
\end{proof}

To recover the unknown part of $G$, we first query the effective resistance for all pairs of vertices with unknown adjacency entries. \cref{thm:unique-graph-completion-k-unknowns} ensures that these $k$ values, together with the partially known adjacency matrix of $G$, uniquely determine the graph. We then search over the set of possible weights for the unknown adjacency entries to find a configuration that is consistent with our ER queries.
\begin{theorem}
    \label{thm:k-unknowns-exponential-completion}
    Let $G=(V,E)$ be a connected weighted graph with weights from a finite known set of size $s$.
    Suppose $k$ entries of the adjacency matrix of $G$ are unknown. 
    There exists an algorithm to reconstruct $G$ (i.e., find the missing $k$ weights) with $k$ ER queries. The running time of the algorithm is $O(s^{k}n^\omega)$, where $\omega$ is the matrix multiplication exponent.
\end{theorem}
\full{
\begin{proof}
    Let $G$ be a graph and let $A_G$ be its partially known adjacency matrix with $k$ missing entries.  Let $I$ be the set of all pairs of indices $i,j\in[n]$ such that $A_G(i,j)$ is known, and let $J$ be the set of all pairs of indices $i,j\in [n]$ such that $A_G(i,j)$ is not known.
    Our algorithm spends $k$ ER queries to obtain the effective resistances for every pair of vertices in $J$.

    Then, our algorithm search for edge weights of the pairs in $J$ by exhaustively trying all $s^k$ possibilities.
    For each candidate set of $k$ weights $W$, our algorithm computes all pairs of effective resistances in the completion of $G$ with $W$, and checks if they are consistent with the $k$ ER values that are queried.  If so, our algorithm returns $W$ as the completion of $G$.  By \cref{thm:unique-graph-completion-k-unknowns}, there exists exactly one set of weights that are consistent with our ER queries, therefore, our algorithm will find them.

    There are $s^k$ possibilities for edge weights in $J$.  For each of them, it takes $O(n^\omega)$ time to compute all effective resistances. %
    Hence, the entire process takes $O(s^kn^{\omega})$ time.
\end{proof}
}

\printbibliography

@article{journals/arxiv/bestide2023reconstnolongcycle,
  author       = {Paul Bastide and
                  Carla Groenland},
  title        = {Optimal distance query reconstruction for graphs without long induced
                  cycles},
  journal      = {CoRR},
  volume       = {abs/2306.05979},
  year         = {2023},
  url          = {https://doi.org/10.48550/arXiv.2306.05979},
  doi          = {10.48550/ARXIV.2306.05979},
  eprinttype    = {arXiv},
  eprint       = {2306.05979},
  bibsource    = {dblp computer science bibliography, https://dblp.org}
}

@article{journals/tcs/WittmannSBT09,
  author       = {Dominik M. Wittmann and
                  Daniel Schmidl and
                  Florian Bl{\"{o}}chl and
                  Fabian J. Theis},
  title        = {Reconstruction of graphs based on random walks},
  journal      = {Theor. Comput. Sci.},
  volume       = {410},
  number       = {38-40},
  pages        = {3826--3838},
  year         = {2009},    
  url          = {https://doi.org/10.1016/j.tcs.2009.05.026},
  doi          = {10.1016/J.TCS.2009.05.026},
  bibsource    = {dblp computer science bibliography, https://dblp.org}
}

@inproceedings{conf/soda/KingZZ03,
  author       = {Valerie King and
                  Li Zhang and
                  Yunhong Zhou},
  title        = {On the complexity of distance-based evolutionary tree reconstruction},
  booktitle    = {Proceedings of the Fourteenth Annual {ACM-SIAM} Symposium on Discrete
                  Algorithms, January 12-14, 2003, Baltimore, Maryland, {USA}},
  pages        = {444--453},
  publisher    = {{ACM/SIAM}},
  year         = {2003},
  url          = {http://dl.acm.org/citation.cfm?id=644108.644179},
  bibsource    = {dblp computer science bibliography, https://dblp.org}
}

@inproceedings{conf/wg/BouvelGK05,
  author       = {Mathilde Bouvel and
                  Vladimir Grebinski and
                  Gregory Kucherov},
  editor       = {Dieter Kratsch},
  title        = {Combinatorial Search on Graphs Motivated by Bioinformatics Applications:
                  {A} Brief Survey},
  booktitle    = {Graph-Theoretic Concepts in Computer Science, 31st International Workshop,
                  {WG} 2005, Metz, France, June 23-25, 2005, Revised Selected Papers},
  series       = {Lecture Notes in Computer Science},
  volume       = {3787},
  pages        = {16--27},
  publisher    = {Springer},
  year         = {2005},
  url          = {https://doi.org/10.1007/11604686\_2},
  doi          = {10.1007/11604686\_2},
  bibsource    = {dblp computer science bibliography, https://dblp.org}
}

@article{journals/siamdm/AlonA05,
  author       = {Noga Alon and
                  Vera Asodi},
  title        = {Learning a Hidden Subgraph},
  journal      = {{SIAM} J. Discret. Math.},
  volume       = {18},
  number       = {4},
  pages        = {697--712},
  year         = {2005},
  url          = {https://doi.org/10.1137/S0895480103431071},
  doi          = {10.1137/S0895480103431071},
  bibsource    = {dblp computer science bibliography, https://dblp.org}
}

@inproceedings{conf/iwpec/Bestide24,
  author       = {Paul Bastide and
                  Carla Groenland},
  editor       = {{\'{E}}douard Bonnet and
                  Pawel Rzazewski},
  title        = {Quasi-Linear Distance Query Reconstruction for Graphs of Bounded Treelength},
  booktitle    = {19th International Symposium on Parameterized and Exact Computation,
                  {IPEC} 2024, September 4-6, 2024, Royal Holloway, University of London,
                  Egham, United Kingdom},
  series       = {LIPIcs},
  volume       = {321},
  pages        = {20:1--20:11},
  publisher    = {Schloss Dagstuhl - Leibniz-Zentrum f{\"{u}}r Informatik},
  year         = {2024},
  url          = {https://doi.org/10.4230/LIPIcs.IPEC.2024.20},
  doi          = {10.4230/LIPICS.IPEC.2024.20},
  bibsource    = {dblp computer science bibliography, https://dblp.org}
}

@article{journals/jcss/AngluinC08,
  author       = {Dana Angluin and
                  Jiang Chen},
  title        = {Learning a hidden graph using O(logn) queries per edge},
  journal      = {J. Comput. Syst. Sci.},
  volume       = {74},
  number       = {4},
  pages        = {546--556},
  year         = {2008},
  url          = {https://doi.org/10.1016/j.jcss.2007.06.006},
  doi          = {10.1016/J.JCSS.2007.06.006},
  bibsource    = {dblp computer science bibliography, https://dblp.org}
}

@inproceedings{alt/ReyzinS07,
  author       = {Lev Reyzin and
                  Nikhil Srivastava},
  editor       = {Marcus Hutter and
                  Rocco A. Servedio and
                  Eiji Takimoto},
  title        = {Learning and Verifying Graphs Using Queries with a Focus on Edge Counting},
  booktitle    = {Algorithmic Learning Theory, 18th International Conference, {ALT}
                  2007, Sendai, Japan, October 1-4, 2007, Proceedings},
  series       = {Lecture Notes in Computer Science},
  volume       = {4754},
  pages        = {285--297},
  publisher    = {Springer},
  year         = {2007},
  url          = {https://doi.org/10.1007/978-3-540-75225-7\_24},
  doi          = {10.1007/978-3-540-75225-7\_24},
  bibsource    = {dblp computer science bibliography, https://dblp.org}
}

@article{sicomp/AlonBKRS04,
  author       = {Noga Alon and
                  Richard Beigel and
                  Simon Kasif and
                  Steven Rudich and
                  Benny Sudakov},
  title        = {Learning a Hidden Matching},
  journal      = {{SIAM} J. Comput.},
  volume       = {33},
  number       = {2},
  pages        = {487--501},
  year         = {2004},
  url          = {https://doi.org/10.1137/S0097539702420139},
  doi          = {10.1137/S0097539702420139},
  bibsource    = {dblp computer science bibliography, https://dblp.org}
}

@inproceedings{Hoskins2018Inferring,
  title={Inferring Networks from Random Walk-Based Node Similarities},
  author={Hoskins, Jeremy and Musco, Cameron and Musco, Christopher and Tsourakakis, Babis},
  booktitle={Advances in Neural Information Processing Systems},
  volume={31},
  year={2018},
  url={https://papers.nips.cc/paper/2018/hash/2f25f6e3.pdf}
}

@misc{Spielman2012TreesRecNotes,
  author       = {Daniel A. Spielman},
  title        = {Trees and Distances},
  year         = {2012},
  note         = {University Lecture},
}

@article{journals/talg/KannanMZ18,
  author       = {Sampath Kannan and
                  Claire Mathieu and
                  Hang Zhou},
  title        = {Graph Reconstruction and Verification},
  journal      = {{ACM} Trans. Algorithms},
  volume       = {14},
  number       = {4},
  pages        = {40:1--40:30},
  year         = {2018},
}

@misc{spielman2019sagt,
    author = {Daniel Spielman},
    title = {Spectral and Algebraic Graph Theory},
    howpublished = {Available at \url{https://cs-www.cs.yale.edu/homes/spielman/sagt/sagt.pdf}},
    year = {2025}
}

@book{books/cu/Goldreich17,
  author       = {Oded Goldreich},
  title        = {Introduction to Property Testing},
  publisher    = {Cambridge University Press},
  year         = {2017},
}

@article{journals/dmtcs/Bodlaender99,
  author       = {Hans L. Bodlaender},
  title        = {A note on domino treewidth},
  journal      = {Discret. Math. Theor. Comput. Sci.},
  volume       = {3},
  number       = {4},
  pages        = {141--150},
  year         = {1999},
  url          = {https://doi.org/10.46298/dmtcs.256},
  doi          = {10.46298/DMTCS.256},
  bibsource    = {dblp computer science bibliography, https://dblp.org}
}

@article{journals/mpcps/Nash-Williams59,
  author = {Nash-Williams, C. St J. A.},
  doi = {10.1017/S0305004100033879},
  issn = {1469-8064},
  journal = {Mathematical Proceedings of the Cambridge Philosophical Society},
  month = {4},
  number = 02,
  numpages = {14},
  pages = {181--194},
  title = {Random Walk and Electric Currents in Networks},
  volume = 55,
  year = 1959
}

@inproceedings{
black2024comparing,
title={Comparing Graph Transformers via Positional Encodings},
author={Mitchell Black and Zhengchao Wan and Gal Mishne and Amir Nayyeri and Yusu Wang},
booktitle={Forty-first International Conference on Machine Learning},
year={2024},
url={https://openreview.net/forum?id=va3r3hSA6n}
}

@inbook{ADWXX2024,
author = {Josh Alman and Ran Duan and Virginia Vassilevska Williams and Yinzhan Xu and Zixuan Xu and Renfei Zhou},
title = {More Asymmetry Yields Faster Matrix Multiplication},
booktitle = {Proceedings of the 2025 Annual ACM-SIAM Symposium on Discrete Algorithms (SODA)},
chapter = {},
year = 2025,
pages = {2005-2039},
doi = {10.1137/1.9781611978322.63},
URL = {https://epubs.siam.org/doi/abs/10.1137/1.9781611978322.63},
eprint = {https://epubs.siam.org/doi/pdf/10.1137/1.9781611978322.63}
}

@book{thomson1867treatise,
  author    = {William Thomson and Peter Tait},
  title     = {Treatise on Natural Philosophy},
  year      = {1867},
  publisher = {Oxford University Press},
  address   = {United Kingdom},
  language  = {English},
  subject   = {Physics},
  genre     = {Non-fiction}
}

@book{rayleigh1877theory,
  author    = {Rayleigh, John William Strutt, Baron},
  title     = {The Theory of Sound},
  volume    = {1},
  year      = {1877},
  publisher = {Macmillan and Co.},
  address   = {London},
  language  = {English},
  topics    = {Sound, Vibration, Waves},
  library   = {University of Michigan},
  collection = {Americana}
}

@article{Griewank1991,
  author    = {Andreas Griewank and Hubertus Th. Jongen and Man Kam Kwong},
  title     = {The equivalence of strict convexity and injectivity of the gradient in bounded level sets},
  journal   = {Mathematical Programming},
  volume    = {51},
  number    = {1},
  pages     = {273--278},
  year      = {1991},
  month     = {July},
  doi       = {10.1007/BF01586939},
  url       = {https://doi.org/10.1007/BF01586939},
  issn      = {1436-4646}
}

@article{vaswani2017attention,
  title={Attention is all you need},
  author={Vaswani, A},
  journal={Advances in Neural Information Processing Systems},
  year={2017}
}

@article{velingker2024affinity,
  title={Affinity-aware graph networks},
  author={Velingker, Ameya and Sinop, Ali and Ktena, Ira and Veli{\v{c}}kovi{\'c}, Petar and Gollapudi, Sreenivas},
  journal={Advances in Neural Information Processing Systems},
  volume={36},
  year={2024}
}

@inproceedings{
xu2018how,
title={How Powerful are Graph Neural Networks?},
author={Keyulu Xu and Weihua Hu and Jure Leskovec and Stefanie Jegelka},
booktitle={International Conference on Learning Representations},
year={2019},
url={https://openreview.net/forum?id=ryGs6iA5Km},
}

@inproceedings{
zhang2023rethinking,
title={Rethinking the Expressive Power of {GNN}s via Graph Biconnectivity},
author={Bohang Zhang and Shengjie Luo and Liwei Wang and Di He},
booktitle={The Eleventh International Conference on Learning Representations },
year={2023},
url={https://openreview.net/forum?id=r9hNv76KoT3}
}

@book{book/affarwal2014social_network_data,
    author = {Aggarwal, Charu C.},
    title = {Social Network Data Analytics},
    year = {2014},
    isbn = {1489988939},
    publisher = {Springer Publishing Company, Incorporated},
}

@book{Magnus_Neudecker_2019, address={Hoboken, NJ}, edition={Third edition}, series={Wiley series in probability and statistics}, title={Matrix differential calculus with applications in statistics and econometrics}, ISBN={9781119541196}, callNumber={QA188}, publisher={Wiley}, author={Magnus, Jan R. and Neudecker, Heinz}, year={2019}, collection={Wiley series in probability and statistics} }

@article{hein1989optimal,
  title={An optimal algorithm to reconstruct trees from additive distance data},
  author={Hein, Jotun J.},
  journal={Bulletin of Mathematical Biology},
  volume={51},
  number={5},
  pages={597--603},
  year={1989},
  publisher={Springer},
  doi={10.1007/BF02459968}
}

@article{journals/ipl/ReyzinS07,
  author       = {Lev Reyzin and
                  Nikhil Srivastava},
  title        = {On the longest path algorithm for reconstructing trees from distance
                  matrices},
  journal      = {Inf. Process. Lett.},
  volume       = {101},
  number       = {3},
  pages        = {98--100},
  year         = {2007},
  url          = {https://doi.org/10.1016/j.ipl.2006.08.013},
  doi          = {10.1016/J.IPL.2006.08.013},
  bibsource    = {dblp computer science bibliography, https://dblp.org}
}

@article{KellerGehrig1985,
title = {Fast algorithms for the characteristics polynomial},
journal = {Theoretical Computer Science},
volume = {36},
pages = {309-317},
year = {1985},
issn = {0304-3975},
doi = {https://doi.org/10.1016/0304-3975(85)90049-0},
url = {https://www.sciencedirect.com/science/article/pii/0304397585900490},
author = {Walter Keller-Gehrig}
}

@article{journals/jsac/BeerlixovaEEHHMR06,
  author       = {Zuzana Beerliova and
                  Felix Eberhard and
                  Thomas Erlebach and
                  Alexander Hall and
                  Michael Hoffmann and
                  Mat{\'{u}}s Mihal{\'{a}}k and
                  L. Shankar Ram},
  title        = {Network Discovery and Verification},
  journal      = {{IEEE} J. Sel. Areas Commun.},
  volume       = {24},
  number       = {12},
  pages        = {2168--2181},
  year         = {2006},
  url          = {https://doi.org/10.1109/JSAC.2006.884015},
  doi          = {10.1109/JSAC.2006.884015},
  bibsource    = {dblp computer science bibliography, https://dblp.org}
}

@article{journals/tcs/RongLYW21,
  author       = {Guozhen Rong and
                  Wenjun Li and
                  Yongjie Yang and
                  Jianxin Wang},
  title        = {Reconstruction and verification of chordal graphs with a distance
                  oracle},
  journal      = {Theor. Comput. Sci.},
  volume       = {859},
  pages        = {48--56},
  year         = {2021},
  url          = {https://doi.org/10.1016/j.tcs.2021.01.006},
  doi          = {10.1016/J.TCS.2021.01.006},
  bibsource    = {dblp computer science bibliography, https://dblp.org}
}

@incollection{Uhlmann2012,
  author    = {Gunther Uhlmann},
  title     = {30 Years of Calderón’s Problem},
  booktitle = {Séminaire Laurent Schwartz — EDP et applications},
  year      = {2012},
  number    = {13},
  pages     = {25},
  doi       = {10.5802/slsedp.40},
  url       = {http://www.numdam.org/articles/10.5802/slsedp.40/}
}

@article{uhlmann2009electrical,
  title={Electrical impedance tomography and Calder{\'o}n's problem},
  author={Uhlmann, Gunther},
  journal={Inverse problems},
  volume={25},
  number={12},
  pages={123011},
  year={2009},
  publisher={IOP Publishing}
}

@misc{wikipediaERT,
  author = "{Wikipedia contributors}",
  title = "{Electrical resistivity tomography --- Wikipedia, The Free Encyclopedia}",
  year = {2024},
  url = {https://en.wikipedia.org/wiki/Electrical_resistivity_tomography},
  note = {Accessed: 2024-01-24}
}

@misc{book/matrix_cookbook,
    author       = "K. B. Petersen and M. S. Pedersen",
    title        = "The Matrix Cookbook",
    year         = "2012",
    month        = "nov",
    publisher    = "Technical University of Denmark",
    address      = "",
    note         = "Version 20121115",
    url          = "http://www2.compute.dtu.dk/pubdb/pubs/3274-full.html"
}

@book{book/boyd-convex-opt, place={Cambridge}, title={Convex Optimization}, publisher={Cambridge University Press}, author={Boyd, Stephen and Vandenberghe, Lieven}, year={2004}}

@inproceedings{conf/stoc/Indyk99,
  author       = {Piotr Indyk},
  editor       = {Jeffrey Scott Vitter and
                  Lawrence L. Larmore and
                  Frank Thomson Leighton},
  title        = {Sublinear Time Algorithms for Metric Space Problems},
  booktitle    = {Proceedings of the Thirty-First Annual {ACM} Symposium on Theory of
                  Computing, May 1-4, 1999, Atlanta, Georgia, {USA}},
  pages        = {428--434},
  publisher    = {{ACM}},
  year         = {1999},
  url          = {https://doi.org/10.1145/301250.301366},
  doi          = {10.1145/301250.301366},
  bibsource    = {dblp computer science bibliography, https://dblp.org}
}

@inproceedings{conf/icalp/Onak08,
  author       = {Krzysztof Onak},
  editor       = {Luca Aceto and
                  Ivan Damg{\aa}rd and
                  Leslie Ann Goldberg and
                  Magn{\'{u}}s M. Halld{\'{o}}rsson and
                  Anna Ing{\'{o}}lfsd{\'{o}}ttir and
                  Igor Walukiewicz},
  title        = {Testing Properties of Sets of Points in Metric Spaces},
  booktitle    = {Automata, Languages and Programming, 35th International Colloquium,
                  {ICALP} 2008, Reykjavik, Iceland, July 7-11, 2008, Proceedings, Part
                  {I:} Tack {A:} Algorithms, Automata, Complexity, and Games},
  series       = {Lecture Notes in Computer Science},
  volume       = {5125},
  pages        = {515--526},
  publisher    = {Springer},
  year         = {2008},
  url          = {https://doi.org/10.1007/978-3-540-70575-8\_42},
  doi          = {10.1007/978-3-540-70575-8\_42},
  bibsource    = {dblp computer science bibliography, https://dblp.org}
}

@article{journals/iandc/ParnasR03,
  author       = {Michal Parnas and
                  Dana Ron},
  title        = {Testing metric properties},
  journal      = {Inf. Comput.},
  volume       = {187},
  number       = {2},
  pages        = {155--195},
  year         = {2003},
  url          = {https://doi.org/10.1016/S0890-5401(03)00160-3},
  doi          = {10.1016/S0890-5401(03)00160-3},
  bibsource    = {dblp computer science bibliography, https://dblp.org}
}

@article{ghosh2008minimizing,
  title={Minimizing effective resistance of a graph},
  author={Ghosh, Arpita and Boyd, Stephen and Saberi, Amin},
  journal={SIAM review},
  volume={50},
  number={1},
  pages={37--66},
  year={2008},
  publisher={SIAM}
}

@MISC{Math-Overflow-Speyer,
    TITLE = {Uniquely reconstruct a matrix $M$ from its inverse $M^{-1}$ if $n$ elements of $M^{-1}$ are unknown and $n$ elements of $M$ are given},
    AUTHOR = {David E Speyer},
    year = {2021},
    HOWPUBLISHED = {MathOverflow},
    URL = {https://mathoverflow.net/q/407563}
}
\end{document}